\documentclass{lmcs}
\pdfoutput=1

\usepackage[utf8]{inputenc}

\usepackage{lastpage}
\lmcsdoi{17}{3}{20}
\lmcsheading{}{\pageref{LastPage}}{}{}%
{Sep.~25,~2020}{Aug.~13,~2021}{}

\usepackage{amsmath}
\usepackage{amssymb}
\usepackage{amsthm}
\usepackage{xspace}
\usepackage{algorithm}
\usepackage{algpseudocode}
\usepackage{float}
\usepackage{xcolor}
\usepackage{tikz}

\usetikzlibrary{patterns}
\usetikzlibrary{fit}
\usetikzlibrary{calc}

\usepackage{thmtools}
\usepackage{thm-restate}

\newcommand{\uppertype}[2]{[#1]^+_{#2}}
\newcommand{\lowertype}[2]{[#1]^-_{#2}}
\newcommand{\fractype}[2]{[#1]_{#2}}
\newcommand{\sifo}{\ensuremath{\textup{Succ-inv }}\FO}
\newcommand{\lsifo}{\ensuremath{\textup{LinSucc-inv }}\FO}
\newcommand{\MC}{\textsc{MC}}
\newcommand{\FO}{\textsc{FO}\xspace}
\newcommand{\LFP}{\textsc{LFP}\xspace}
\newcommand{\MSO}{\textsc{MSO}\xspace}
\newcommand{\oifo}{\ensuremath{<\!\textup{-inv }}\FO}

\newcommand{\sieq}[1]{\equiv^{{\sifo}}_{#1}}

\newcommand{\foeq}[1]{\equiv^{\FO}_{#1}}
\newcommand{\N}{\mathbb{N}}
\newcommand{\dist}[3]{\text{dist}_{#1}(#2,#3)}
\newcommand{\tp}[3]{\text{tp}_{#2}^{#3}(#1)}

\newcommand{\ntp}[1]{[\![#1]\!]}
\newcommand{\nocc}[2]{|#1|_{#2}}
\newcommand{\im}{\operatorname{Im}}
\newcommand{\neigh}[3]{\mathcal N_{#2}^{#3}(#1)}

\newcommand{\lay}[1]{\texttt{(Layer[$#1$])}\xspace}
\newcommand{\cone}[1]{C_{#1}}
\newcommand{\threq}[4]{\ntp{#1}_{#3}=^{#4}\ntp{#2}_{#3}}

\def\eg{{\em e.g.}}
\def\ie{{\em i.e.}\xspace}

\title{Successor-Invariant First-Order Logic on Classes of Bounded Degree}

\author{Julien Grange}
\address{ENS Paris, PSL, INRIA, CNRS}
\email{julien.grange@inria.fr}

\begin{document}

\begin{abstract}
  We study the expressive power of successor-invariant first-order logic, which is an extension of first-order logic where the usage of an additional successor relation on the structure is allowed, as long as the validity of formulas is independent of the choice of a particular successor on finite structures.

  We show that when the degree is bounded, successor-invariant first-order logic is no more expressive than first-order logic.
\end{abstract}

\maketitle

\section{Introduction}

First-order logic, \FO, is the standard formalism to express properties of finite structures. Its expressive power is well known, and very restrained, as it can only express properties that are local, which roughly means that it can only talk about the immediate surroundings of a small number of elements, and it is unable to count.

A number of logics with higher expressivity can be defined with \FO as a building block, such as \MSO, in which quantification over sets is allowed, and \LFP, which adds a fixpoint operator to \FO. These additions break the local character of the logic. 

Another way to define logics from \FO is through the addition, in an invariant way, of arithmetic predicates on the structure that are exterior to the vocabulary. This amounts to arbitrarily identifying the universe of the structure with an initial segment of the integers, and allowing some arithmetic on them. However, we want these extensions to define properties of the structures, and not to depend on a particular ordering on their elements: thus we focus on invariant extensions of \FO.

If the only predicate allowed is the order, we get order-invariant first-order logic, \oifo. Restricting a bit the additional relation, we get successor-invariant first-order, \sifo. In this formalism, we only grant access to the successor relation derived from the order, provided that the evaluation of a sentence using this successor relation is independent of the choice of a particular successor.

The problem of determining whether an \FO-sentence using an order or a successor relation is invariant wrt. this relation is undecidable, by reduction from the finite satisfiability problem (cf. Trakhtenbrot's theorem \cite{trakhtenbrot1950impossibility}). In fact, it is shown in~\cite{DBLP:journals/jsyml/BenediktS09} that this problem in undecidable even on strings. Hence we use here the term ``logic'' somewhat liberally, since having a recursively enumerable syntax is a usual requirement for a logic. 

The study of these two formalisms finds its motivation, among other topics such as descriptive complexity, in database theory. As databases are commonly stored on disk that implicitly order their memory segments, when one wishes to express a query in \FO, one has access to an additional order on the elements of the database. However, making use of this order without care could result in queries that evaluate differently on two implementations of the same database, which is clearly an undesirable behavior. We want to use this order only in an invariant fashion; this way, the result of a query depends only on the database it is run on, and not on the way the data is stored on disk. This amounts exactly to the definition of \oifo, or \sifo if we restrict the way this order can be accessed.

It is straightforward that \oifo is at least as expressive as \sifo, which in turn can express any \FO-definable property. Gurevich constructed a class of finite structures that can be defined by an \oifo sentence, but which is not \FO-definable. Though this construction was not published by Gurevich, it can be found, \eg, in Section 5.2 of \cite{DBLP:books/sp/Libkin04}. Rossman extended this result, and proved in \cite{DBLP:journals/jsyml/Rossman07} that on finite structures, \sifo is strictly more expressive than \FO.

Grohe and Schwentick \cite{DBLP:journals/tocl/GroheS00} proved that these logics were Gaifman-local, giving an upper bound to their expressive power. Other upper bounds were given by Benedikt and Segoufin \cite{DBLP:journals/jsyml/BenediktS09}, who proved that \oifo, and hence \sifo, are included in \MSO on classes of bounded treewidth and on classes of bounded degree. Elberfeld, Frickenschmidt and Grohe \cite{DBLP:conf/lics/ElberfeldFG16} extended the first inclusion to a broader setting, that of decomposable structures. Whether these logics are included in \MSO in general is still an open question.

 The classes of structures involved in the separating examples by Gurevich and Rossman are dense, and no other example is known on classes that are sparse. Far from it, \oifo and \emph{a fortiori} \sifo are known to collapse to \FO on several sparse classes. Benedikt and Segoufin \cite{DBLP:journals/jsyml/BenediktS09} proved the collapse on trees; Eickmeyer, Elberfeld and Harwarth \cite{DBLP:conf/mfcs/EickmeyerEH14} obtained an analogous result on graphs of bounded tree-depth; Grange and Segoufin \cite{grange_et_al:LIPIcs:2020:11666} proved the collapse on hollow trees.

Whether \oifo or \sifo collapse to \FO on classes of graphs of bounded treewidth (or even bounded pathwidth) are still open questions. We go in another direction in this paper, and prove that \sifo collapses to \FO on classes of structures of bounded degree. To do this, we show how to construct successors on two \FO-similar structures of bounded degree, such that the two structures remain \FO-similar when considering the additional successor relation.

\paragraph{\textbf{Related work:}} The general method used in \cite{DBLP:conf/mfcs/EickmeyerEH14} to prove that \oifo collapses to \FO when the tree-depth is bounded is the same as ours: starting from two \FO-similar structures, they show how to construct orders that maintain the similarity. However, the techniques we use to construct our successors are nothing like the ones used in \cite{DBLP:conf/mfcs/EickmeyerEH14}, as the assumptions on the classes at hand (bounded tree-depth versus bounded degree) are very different.

Instead of directly constructing similar orders on two similar structures, \cite{DBLP:journals/jsyml/BenediktS09} and \cite{grange_et_al:LIPIcs:2020:11666} exhibit a chain of intermediate structures and intermediate orders that are pairwise similar, in order to prove the collapse on trees and hollow trees. Although all these constructions, as well as ours, rely on a careful manipulation of the neighborhoods, our construction differs widely from these ones. Indeed, instead of chaining local modifications of the structures, we construct all at once our successor relations, without intermediate steps.

The classes of graphs on which the model checking problem for \FO (denoted \MC(\FO)) is fixed-parameter tractable has been widely studied. It has originally been proven by Seese~\cite{DBLP:journals/mscs/Seese96} that \MC(\FO) is fixed-parameter linear on any class of bounded degree. After a series of improvements on this result, Grohe, Kreutzer and Siebertz~\cite{DBLP:journals/jacm/GroheKS17} showed that this problem is fixed-parameter tractable on any nowhere dense class of graphs.

Concerning the model checking problem for \sifo, Van den Heuvel, Kreutzer, Pilipczuk, Quiroz, Rabinovich and Siebertz~\cite{DBLP:conf/lics/HeuvelKPQRS17} proved that \MC(\sifo) is fixed-parameter tractable on any class of bounded expansion (which is less general than the nowhere dense setting, but also includes any class of bounded degree). Since there is no indication that \sifo is more expressive than \FO on classes of bounded expansion, this could possibly be due to a collapse of \sifo to \FO on those classes. Our result showing that \sifo collapses to \FO on classes of bounded degree, together with the aforementioned result from \cite{DBLP:journals/mscs/Seese96}, gives an alternative proof of the fact that \MC(\sifo) is non-uniform fixed-parameter linear when the degree is bounded.

\section{Preliminaries}
\label{sec:prelim}

The remainder in the division of $n\in\N$ by $m>0$ is denoted $n[m]$.

A binary relation on a finite set $X$ is a \textbf{successor relation on $X$} if it is the graph of a circular permutation of $X$, \ie a bijective function from $X$ to $X$ with a single orbit. This differs from the standard notion of successor in that there is neither minimal nor maximal element. However, this does not have any impact on our result, as discussed at the end of the present section.

We use the standard definition of first-order logic $\FO(\Sigma)$ over a signature $\Sigma$ composed of relation and constant symbols. We only consider finite $\Sigma$-structures, which are denoted by calligraphic upper-case letters, while their universes are denoted by the corresponding standard upper-case letters; for instance, $A$ is the universe of the structure $\mathcal A$. 

\begin{defi}[\sifo]
  A sentence $\varphi\in\FO(\Sigma\cup\{S\})$, where $S$ is a binary
  relation symbol, is said to be \textbf{successor-invariant} if for every
  $\Sigma$-structure $\mathcal A$, and every successor relations $S_1$ and
  $S_2$ on $A$, $(\mathcal A,S_1)\models\varphi$ iff
  $(\mathcal A,S_2)\models\varphi$. We can then omit the interpretation for
  the symbol $S$, and if $(\mathcal A,S_1)\models\varphi$ for any (every)
  successor $S_1$, we write $\mathcal A\models\varphi$. 

  The set of  successor-invariant sentences on $\Sigma$ is denoted $\sifo(\Sigma)$.
\end{defi}

\begin{defi}[$\mathcal L$-similarity]
  Given two $\Sigma$-structures $\mathcal A$ and $\mathcal B$, and $\mathcal L$ being either $\FO(\Sigma)$ or $\sifo(\Sigma)$, we write $\mathcal A\equiv^{\mathcal L}_k\mathcal B$, and say that that $\mathcal A$ and $\mathcal B$ are \textbf{$\mathcal L$-similar at depth $k$}, if $\mathcal A$ and $\mathcal B$ satisfy the same $\mathcal L$-sentences of quantifier rank at most $k$. For $\FO(\Sigma)$ as well as for $\sifo(\Sigma)$, we omit $\Sigma$ when it is clear from the context.
\end{defi}

We write $\mathcal A\simeq\mathcal B$ if $\mathcal A$ and $\mathcal B$ are isomorphic.

\begin{defi}[Gaifman graph]
  The Gaifman graph $\mathcal G_{\mathcal A}$ of a $\Sigma$-structure
  $\mathcal A$ is defined as $(A,V)$ where $(x,y)\in V$ iff $x$ and $y$ are distinct and
  appear in the same tuple of a relation of $\mathcal A$. In particular, if
  a graph is seen as a relational structure on the vocabulary $\{E\}$, its
  Gaifman graph is the unoriented version of itself. By
  $\dist{\mathcal A}{x}{y}$, we denote the distance between $x$ and $y$ in
  $\mathcal G_{\mathcal A}$. The \textbf{degree} of $\mathcal A$ is the
  degree of its Gaifman graph, and a class $\mathcal C$ of
  $\Sigma$-structures is said to be of \textbf{bounded degree} if there
  exists some $d\in\N$ such that the degree of every
  $\mathcal A\in\mathcal C$ is at most $d$.
\end{defi}

\begin{defi}[Neighborhood types]
  \label{def:types}
  Let $c$ be a constant symbol that does not appear in~$\Sigma$. 
  
  For $k\in\N$ and $x\in A$, the \textbf{$k$-neighborhood}
  $\neigh{x}{\mathcal A}{k}$ of $x$ is the $(\Sigma\cup\{c\})$-structure
  whose $\Sigma$-restriction is the substructure of $\mathcal A$ induced by
  $\{y\in A:\dist{\mathcal A}{x}{y}\leq k\}$, and where $c$ is interpreted
  as $x$. 

  The \textbf{$k$-neighborhood type} $\tau=\tp{x}{\mathcal A}{k}$ is the
  isomorphism class of its $k$-neighborhood. We say that $\tau$ is a neighborhood type
  over $\Sigma$, and that $x$ is an \textbf{occurrence} of
  $\tau$. We denote by $\nocc{\mathcal A}{\tau}$ the number of occurrences of
  $\tau$ in $\mathcal A$, and we write
  $\threq{\mathcal A}{\mathcal B}{k}{t}$ to mean that for every $k$-neighborhood type
  $\tau$, $\nocc{\mathcal A}{\tau}$ and $\nocc{\mathcal B}{\tau}$ are
  either equal, or both larger than $t$.
\end{defi}

\begin{defi}[Path and cycles]
  A \textbf{cycle} of length $l\geq 3$ in the $\Sigma\cup\{S\}$-structure
  $\mathcal A$ is a sequence $(x_0,\dots,x_{l-1})$ of distinct vertices of
  $A$ such that for every $0\leq i< l$, $x_i$ and $x_{i+1[l]}$ appear in
  the same tuple of some relation of $\mathcal A$ (in other words, it is a
  cycle in $\mathcal G_{\mathcal A}$). If furthermore
  $(x_i,x_{i+1[l]})\in S$ for every $i$, then we say that it is an
  $S$-cycle. If for some $i$, $(x_i,x_{i+1[l]})\in S$ or
  $(x_{i+1[l]},x_i)\in S$, then we say that the cycle goes through an
  $S$-edge. A \textbf{path} is defined similarly, without the requirement
  on $x_{l-1}$ and $x_0$, and its length is $l-1$ instead of $l$.
\end{defi}

From now on, we assume that $\Sigma$ is purely relational (\ie contains only relation symbols) and does not contain the binary symbol $S$.

\medskip

For a class $\mathcal C$ of $\Sigma$-structures, we say that \[\sifo=\FO\text{ on }\mathcal C\] if the properties of $\mathcal C$ definable in $\sifo$ and in $\FO$ are the same. In other words, if for every $\varphi\in\sifo$, there exists some $\bar\varphi\in\FO$ such that \[\forall\mathcal A\in\mathcal C,\quad\mathcal A\models\varphi\quad\text{iff}\quad\mathcal A\models\bar\varphi\,.\] The reverse inclusion, \ie $\FO\subseteq\sifo$, always holds and needs no verification.

We are now ready to state our main result:

\begin{restatable}{thm}{collapse}
  \label{th:collapse}
  For every vocabulary $\Sigma$ and for every class $\mathcal C$ of $\Sigma$-structures of bounded degree, 
  \[\sifo=\FO\text{ on }\mathcal C\,.\]
\end{restatable}

The proof of Theorem~\ref{th:collapse} is given in Section~\ref{sec:proof}, and constitutes the core of this paper. We give here a sketch of this proof; this will motivate the definitions given in Section~\ref{sec:fractal}.

\begin{proof}[Proof overview]
  Our goal is, given two structures $\mathcal G_1$ and $\mathcal G_2$ of degree at most $d$ that are \FO-similar (that is, such that $\mathcal G_1\foeq{n}\mathcal G_2$ for a large enough $n$), to construct a successor relation $S_1$ on $\mathcal G_1$ and $S_2$ on $\mathcal G_2$ such that $(\mathcal G_1,S_1)$ and $(\mathcal G_2,S_2)$ stay \FO-similar. We will see that this entails that $\foeq{}$ refines $\sieq{}$ when the degree is bounded. From there, a standard finite-model-theoretic argument (namely, that $\foeq{n}$ has finite index and that each one of its classes is \FO-definable) gives the inclusion $\sifo\subseteq\FO$ on classes of bounded degree.

  It thus remains to construct suitable successor relations $S_1$ and $S_2$. First, we separate the neighborhood types occurring in $\mathcal G_1$ and $\mathcal G_2$ into two categories:

  \begin{itemize}
  \item on the one hand, the rare neighborhood types, which have few occurrences in $\mathcal G_1$ and $\mathcal G_2$ (and thus, that have the same number of occurrences in both structures, by \FO-similarity)
  \item on the other hand, the frequent neighborhood types, which have many occurrences both in $\mathcal G_1$ and $\mathcal G_2$.
  \end{itemize}
  In order to make the proof of \FO-similarity of $(\mathcal G_1,S_1)$ and $(\mathcal G_2,S_2)$ as simple as possible, we want an element of $\mathcal G_1$ and its successor by $S_1$ (and similarly for $\mathcal G_2$ and $S_2$) to have the same neighborhood type in $\mathcal G_1$ as much as possible, and to be far away in $\mathcal G_1$, in order for the neighborhood types occurring in $(\mathcal G_1,S_1)$ to be as ``regular'' as possible. As long as there are at least two different neighborhood types, the first constraint obviously cannot be satisfied, but we will construct $S_1$ as close as possible to satisfying it.

  \medskip

  For instance, suppose that $\mathcal G_1$ contains three frequent neighborhood types $\tau_0,\tau_1$ and $\tau_2$, and one rare neighborhood type $\chi$ with two occurrences. At the end of the construction, $S_1$ will (mostly) look like in Figure~\ref{fig:ex}, where the relations of $\mathcal G_1$ have been omitted and the arrows represent $S_1$, which is indeed a circular successor. 

  Note that all the elements of neighborhood type $\tau_1$ form a segment wrt. $S_1$, as well as all the elements of neighborhood type $\tau_2$. The first frequent neighborhood type, $\tau_0$, has a special role in that it is used to embed all the elements of rare neighborhood type (here, $\chi$). Furthermore, and this is not apparent in the figure, two successive elements for $S_1$ are always distant in $\mathcal G_1$.

  \newcommand{\outr}{3.9}
  \newcommand{\midr}{3.5}
  \newcommand{\inr}{3.2}
    
  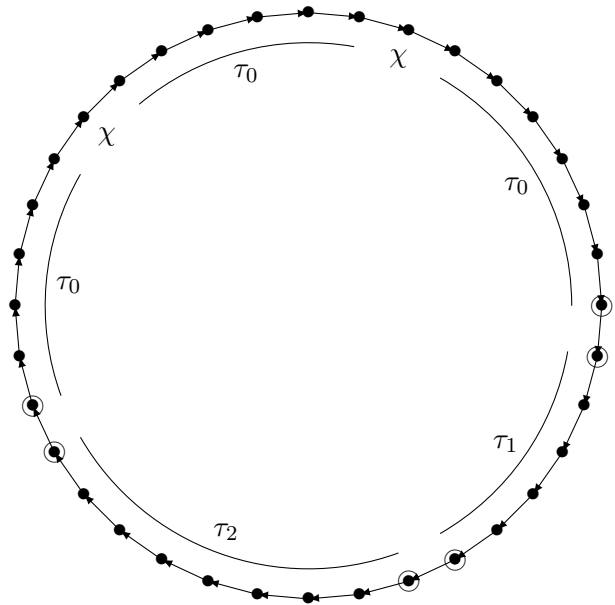
\begin{figure}[ht]
    \centering 
    \begin{tikzpicture}[baseline=(current bounding box.center)]
      
      \coordinate (c) at (0,0);
      
      \foreach \a in {0,...,35}{
        \coordinate (\a) at c++({\a*10}:\outr);
        \coordinate (m\a) at c++({\a*10}:\midr);
      }

      \foreach \a in {0,...,35}
      \draw[->,>=latex] let \n1={int(mod(\a+1,36))} in (\n1) -- (\a);
      
      \draw (m0) arc (0:60:\midr);
      \draw (m8) arc (80:130:\midr);
      \draw (m15) arc (150:200:\midr);
      \draw (m21) arc (210:290:\midr);
      \draw (m30) arc (300:350:\midr);

      \foreach \a in {3,10.5,17.5}{
        \coordinate (i) at c++({\a*10}:\inr);
        \draw (i) node{$\tau_0$};
      }

      \foreach \a in {7,14}{
        \coordinate (i) at c++({\a*10}:\midr);
        \draw (i) node{$\chi$};
      }

      \coordinate (i1) at c++(325:\inr);
      \draw (i1) node{$\tau_1$};

      \coordinate (i2) at c++(250:\inr);
      \draw (i2) node{$\tau_2$};
      
      \foreach \a in {0,...,35}
      \draw (\a) node{$\bullet$};

      \foreach \a in {0,20,21,29,30,35}
      \draw (\a) node{\large $\circledcirc$};

    \end{tikzpicture}
    \caption{Illustration of $S_1$ when there are three frequent neighborhood types ($\tau_0$, $\tau_1$, $\tau_2$) and one rare neighborhood type ($\chi$) in $\mathcal G_1$. The elements of rare neighborhood type are surrounded by occurrences of the first frequent neighborhood type, $\tau_0$. Junction elements are circled.}
    \label{fig:ex}
  \end{figure}

  \medskip
  
  Keeping this idea in mind, $S_1$ (and similarly, $S_2$) is constructed iteratively, by adding $S$-edges to the initial structures one at a time. For practical reasons, we will start the construction of $S_1$ around occurrences of rare neighborhood types: for each element $x$ of rare neighborhood type, we find two elements of neighborhood type $\tau_0$ that are far apart in $\mathcal G_1$, and far from $x$. Then we add two $S$-edges in order for those two elements to become the $S_1$-predecessor and the $S_1$-successor of $x$. We repeat this process for every element of rare neighborhood type (and actually, for every element that belongs to the neighborhood of a rare element) until each one is protected by a ball of elements of frequent neighborhood type. This is possible because there are few elements of rare neighborhood type, and many elements of any frequent neighborhood type; since the degree is bounded, those elements of frequent neighborhood type are spread across the structure, and can be found far from the current construction.

  Once this is done, we apply a similar construction around elements of frequent neighborhood types that will, in the end, be the $S_1$-predecessor or $S_1$-successor of an element of another frequent neighborhood type - that is, elements that will be at the border of the segments (for $S_1$) of a given frequent neighborhood type. Such elements are circled in Figure~\ref{fig:ex}. We must choose only a small number of such elements (two for each frequent neighborhood type, of which there are few due to the degree boundedness hypothesis), hence we can find enough far-apart elements of frequent neighborhood type to embed them. Once again, degree boundedness is crucial.

  After these two steps, $S_1$ has been constructed around all the singular points. It only remains to complete $S_1$ by adding edges between the remaining elements (all of which are occurrences of frequent neighborhood types), in such a way that elements of a same frequent neighborhood type end up forming a segment for $S_1$, and such that $S_1$ brings together elements that were far apart in the initial structure $\mathcal G_1$. Once again, the high number of occurrences of each frequent neighborhood type allows us to do so.

  \medskip

  Applying the same construction to $\mathcal G_2$, we end up with two structures $(\mathcal G_1,S_1)$ and $(\mathcal G_2,S_2)$ that cannot be distinguished by \FO-formulas of small (wrt. the initial \FO-similarity index between $\mathcal G_1$ and $\mathcal G_2$) quantifier rank, which concludes the proof.
  
  We have given a global overview of the construction process of $S_1$; however, there are technical difficulties to take care of, which are dealt with in Section~\ref{sec:proof}. For that, we need the definitions given in Section~\ref{sec:fractal}, which formalize the notion of regularity of a neighborhood type in $(\mathcal G_1,S_1)$ and $(\mathcal G_2,S_2)$.
\end{proof}

Let us now prove that our decision to consider circular successors instead of the more traditional linear ones (with a minimal and a maximal element) bears no consequence on this result. If we define $\lsifo$ in the same way as $\sifo$, but where the invariant relation is a linear successor $\bar S$, we get:

\begin{lem}
  For every vocabulary $\Sigma$, $\lsifo$ and $\sifo$ define the same properties of $\Sigma$-structures.
\end{lem}

\begin{proof}
  Given $\varphi\in\sifo$, let us prove that there exists a formula $\psi\in\lsifo$ such that $\psi$ is equivalent to $\varphi$ (\ie for every $\Sigma$-structure $\mathcal A$, $\mathcal A\models\varphi$ iff $\mathcal A\models\psi$).
  
  Let $\psi$ be defined as $\varphi$ in which every atom $S(x,y)$ has been replaced with the formula $\bar S(x,y)\vee \neg\exists z (\bar S(x,z)\vee \bar S(z,y))$.

  Let $\mathcal A$ be a $\Sigma$-structure and $\bar S$ be a linear successor on $A$. Then $(\mathcal A,\bar S)\models\psi$ iff $(\mathcal A, S)\models\varphi$, where $S$ is the circular successor obtained from $\bar S$ by adding an edge from the maximal element to the minimal one.

  This guarantees that $\psi\in\lsifo$, and that $\psi$ and $\varphi$ are equivalent.

  \medskip

  Conversely, let $\psi\in\lsifo$ and let $\varphi$ be the formula $\exists\min\,\text{Cut}(\psi)$, where $\text{Cut}(\psi)$ is obtained by replacing in $\psi$ every $\bar S(x,y)$ with $S(x,y)\wedge\neg y=\min$.

  Let $\mathcal A$ be a $\Sigma$-structure, let $S$ be a circular successor on $A$, and let $\min\in A$. Then $(\mathcal A, S,\min)\models\text{Cut}(\psi)$ iff $(\mathcal A, \bar S)\models\psi$, where $\bar S$ is the linear successor obtained from $S$ by removing the edge pointing to $\min$. Hence $(\mathcal A,S)\models\varphi$ iff there exists a linear successor $\bar S$ obtained from $S$ by an edge removal such that $(\mathcal A,\bar S)\models\psi$, that is iff $\mathcal A\models\psi$.

  This ensures that $\varphi\in\sifo$ and that $\varphi$ and $\psi$ are equivalent.  
\end{proof}

\section{Fractal types and layering}
\label{sec:fractal}

To prove Theorem~\ref{th:collapse}, we will start from two structures $\mathcal G_1$ and $\mathcal G_2$ that are \FO-similar, and construct successor relations $S_1$ and $S_2$ on their universes so that the structures remain \FO-similar when we take into account these additional successor relations.

We want to construct $S_\epsilon$, for $\epsilon\in\{1,2\}$, in a way that makes $\tp{a}{(\mathcal G_\epsilon,S_\epsilon)}{k}$ as regular as possible for every $a\in G_\epsilon$, in order to ease the proof of \FO-similarity of $(\mathcal G_1,S_1)$ and $(\mathcal G_2,S_2)$.

Ideally, the $S_\epsilon$-successors and $S_\epsilon$-predecessors of any element should have the same $k$-neighborhood type in $\mathcal G_\epsilon$ as this element. On top of that, there should not be any overlap between the $k$-neighborhoods in $\mathcal G_\epsilon$ of elements that are brought closer by $S_\epsilon$ (this ``independence'' is captured by the layering property, introduced in Definition~\ref{def:layer}).

\bigskip

Let us now try to visualize what $\tp{a}{(\mathcal G_\epsilon,S_\epsilon)}{k}$ would look like in those perfect conditions, for some $a$ of $k$-neighborhood type $\tau$ in $\mathcal G_\epsilon$, with Figure~\ref{fig:fractal} as a visual aid.

Let $a^+$ be the successor of $a$ by $S_\epsilon$: in the $k$-neighborhood of $a$ in $(\mathcal G_\epsilon,S_\epsilon)$ appears the $(k-1)$-neighborhood of $a^+$ in $\mathcal G_\epsilon$. But in these ideal conditions, we have that $\tp{a^+}{\mathcal G_\epsilon}{k}=\tp{a}{\mathcal G_\epsilon}{k}$, hence in $\tp{a}{(\mathcal G_\epsilon,S_\epsilon)}{k}$, we see that the pattern $\tp{a}{\mathcal G_\epsilon}{k}$ is repeated around the $S_\epsilon$-successor of $a$, with a radius shrunk by one. If we follow again $S_\epsilon$, the same neighborhood type will appear once more, this time with radius $k-2$, and so on.

Let us now take a step sideways in the $k$-neighborhood of $a$ in $\mathcal G_\epsilon$ (i.e. in the horizontal plane in the figure), and consider some element $x$ at distance $d$ from $a$ in $\mathcal G_\epsilon$. In these perfect conditions, $x$ and its $S_\epsilon$-successor have the same $k$-neighborhood type in $\mathcal G_\epsilon$. Of course, only a part of these neighborhoods will appear in the $k$ neighborhood of $a$ in $(\mathcal G_\epsilon,S_\epsilon)$; namely, the $(k-d)$-neighborhood of $x$ in $\mathcal G_\epsilon$ and the $(k-d-1)$-neighborhood of $x^+$ in $\mathcal G_\epsilon$. If we move to the $S_\epsilon$-successor of $x^+$, we will find that the visible part in $\tp{a}{(\mathcal G_\epsilon,S_\epsilon)}{k}$ of its neighborhood in $\mathcal G_\epsilon$ will have the same $(k-d-2)$-radius type as $x^+$.

Of course, everything we have said about $S_\epsilon$-successors also holds for $S_\epsilon$-predecessors, and for any iteration of upward/downward (i.e. along $S_\epsilon$-edges) and sideways (i.e. in $\mathcal G_\epsilon$) steps, thus encapsulating $a$ in a very regular neighborhood of radius $k$ in $(\mathcal G_\epsilon,S_\epsilon)$.

It should now be visible that this construction is reminiscent of a fractal, in that some patterns - here, the neighborhood types - are repeated each time we follow the successor relation, their size shrinking time after time (although in our setting, the patterns are obviously repeated only a finite number of times).

This is why we introduce in Definition~\ref{def:fractal_types} the \textbf{fractal type} $\fractype{\tau}{k}$.

Aside from a small number of exceptions (namely, for neighborhood types that do not occur frequently enough, and around the transitions between frequent neighborhood types), every element of $k$-neighborhood type $\tau$ in $\mathcal G_\epsilon$ will have the fractal neighborhood type $\fractype{\tau}{k}$ in $(\mathcal G_\epsilon,S_\epsilon)$.

\medskip
If $\mathcal N$ is a representative of a neighborhood type $\tau$, $c^{\mathcal N}$ is called the \textbf{center of $\mathcal N$}. Recall from Definition~\ref{def:types} that $c$ is the constant symbol added to $\Sigma$ when considering neighborhood types to pinpoint the central element of a neighborhood.

\begin{defi}[Fractal types]
  \label{def:fractal_types}
  We define by induction on $k\in\N$, for every $k$-neighborhood type $\tau$ over $\Sigma$, the $k$-neighborhood types $\fractype{\tau}{k}$, $\uppertype{\tau}{k}$ and $\lowertype{\tau}{k}$ over $\Sigma\cup\{S\}$.
  
  For $k=0$, $\fractype{\tau}{0}=\uppertype{\tau}{0}=\lowertype{\tau}{0}=\tau$ (meaning that $S$ is interpreted as the empty relation in $\fractype{\tau}{0}$, $\uppertype{\tau}{0}$ and $\lowertype{\tau}{0}$).

  Starting from a representative $\mathcal N$ of center $a$ of the isomorphism class $\tau$, we construct $\mathcal N'$, whose role is to serve as an intermediate step in the construction of the full fractal neighborhood, as follows.

  For every $x\in N$ at distance $d\leq k-1$ from $a$, let $\mathcal M^+_x$ and $\mathcal M^-_x$ be structures of respective isomorphism type $\uppertype{\chi}{k-d-1}$ and $\lowertype{\chi}{k-d-1}$, where $\chi$ is the $(k-d-1)$-neighborhood type of $x$ in $\mathcal N$, and of respective center $x^+$ and $x^-$.

  $\mathcal N'$ is defined as the disjoint union of $\mathcal N$ and all the $\mathcal M^+_x$ and the $\mathcal M^-_x$, for $x\neq a$, together with all the edges $S(x,x^+)$ and $S(x^-,x)$.
  
  From there, $\mathcal N^+$ (resp. $\mathcal N^-$) is defined as the the disjoint union of $\mathcal N'$ and $\mathcal M^+_a$ (resp. $\mathcal M^-_a$) together with the edge $S(a,a^+)$ (resp. $S(a^-,a)$). Likewise, $\mathcal N^{+/-}$ is defined as the disjoint union of $\mathcal N'$, $\mathcal M^+_a$ and $\mathcal M^-_a$ together with the edges $S(a,a^+)$ and $S(a^-,a)$. In each case, $a$ is taken as the center.

  Now, $\fractype{\tau}{k}$, $\uppertype{\tau}{k}$ and $\lowertype{\tau}{k}$ are defined respectively as the isomorphism type of $\mathcal N^{+/-}$, $\mathcal N^+$ and $\mathcal N^-$.

  An illustration of this definition is given in Figure~\ref{fig:fractal}.

\end{defi}

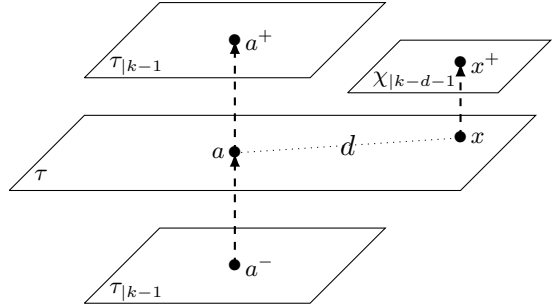
\begin{figure}[!ht]
  \centering
  \begin{tikzpicture}[baseline=(current bounding box.center)]
    
    \coordinate (1) at (0,0);
    \coordinate (2) at (0,1.5);
    \coordinate (3) at (0,-1.5);
    \coordinate (4) at (3,0.2);
    \coordinate (5) at (3,1.2);
    
    \draw (-3,-0.5) -- (3,-0.5) -- (4,0.5) -- (-2,0.5) -- (-3,-0.5);
    \draw (-2,1) -- (1,1) -- (2,2) -- (-1,2) -- (-2,1);
    \draw (-2,-2) -- (1,-2) -- (2,-1) -- (-1,-1) -- (-2,-2);

    \draw (1.5,0.8) -- (3.5,0.8) -- (4.2,1.5) -- (2.2,1.5) -- (1.5,0.8);

    \draw (1) node[left] {\footnotesize $a$} node {$\bullet$};
    \draw (2) node[right] {\footnotesize $a^+$} node {$\bullet$};
    \draw (3) node[right] {\footnotesize $a^-$} node {$\bullet$};
    \draw (4) node[right] {\footnotesize $x$} node {$\bullet$};
    \draw (5) node[right] {\footnotesize $x^+$} node {$\bullet$};

    \draw (-2.8,-0.5) node[above right] {\footnotesize $\tau$} ;
    \draw (-1.8,0.9) node[above right] {\footnotesize $\tau_{|k-1}$} ;
    \draw (-1.8,-2.1) node[above right] {\footnotesize $\tau_{|k-1}$} ;
    \draw (1.7,0.7) node[above right] {\footnotesize $\chi_{|k-d-1}$} ;
    
    \foreach \from/\to in {1/2,3/1,4/5}
    \draw[->,thick,dashed,>=latex] (\from) to (\to);

    \draw[dotted] (1) to node[midway]{$d$} (4);
        
  \end{tikzpicture}
  \caption{Partial representation of $\mathcal N^{+/-}$, of type $\fractype{\tau}{k}$. Here, $\chi$ is the $(k-d)$-neighborhood type of the element $x$, at distance $d$ from $a$ in $\tau$. The dashed arrows represent $S$-edges.}
  \label{fig:fractal}
\end{figure}

\begin{defi}[Layering]
  \label{def:layer}
  We say that an $r$-neighborhood $\mathcal N$ over $\Sigma\cup\{S,c\}$ is \textbf{layered} if it does not contain any cycle going through an $S$-edge. Every $\fractype{\tau}{r}$ is obviously layered by construction.

  We say that a structure over $\Sigma\cup\{S\}$ satisfies the property \lay{r} iff all the $r$-neighborhoods of this structure are layered. 

\end{defi}

It turns out \lay{r} can be reformulated in a way that does not involve the $r$-neighborhoods of the structure.

\begin{lem}
  A structure $\mathcal G$ over $\Sigma\cup\{S\}$ satisfies \lay{r} if and only if it contains no cycle of length at most $2r+1$ going through an $S$-edge.
\end{lem}

\begin{proof}
  If $\mathcal G$ contains a cycle of length at most $2r+1$ going through an $S$-edge, then the $r$-neighborhood of any vertex of this cycle contains the whole cycle, thus \lay{r} does not hold in $\mathcal G$.

  Conversely, we show that if some $r$-neighborhood contains a cycle going through an $S$-edge, then it must also contain a small (i.e. of length at most $2r+1$) such cycle.

  Suppose that there exists $x\in G$ such that $\neigh{x}{\mathcal G}{r}$ contains a cycle going through an $S$-edge, and let $S(y,z)$ be such an edge.

  For any $u\in\neigh{x}{\mathcal G}{r}$, we define the \textbf{cone $\cone{u}$ at $u$} as the set of elements $v\in\neigh{x}{\mathcal G}{r}$ such that every shortest path from $x$ to $v$ in $\neigh{x}{\mathcal G}{r}$ goes through $u$.
  
  There are two cases, depending on the relative position of $y$, $z$ and their cones:
  
  \begin{itemize}
  \item If $z\notin\cone{y}$ and $y\notin\cone{z}$, let $p_{y\rightarrow x}$ (resp. $p_{x\rightarrow z}$) be a path of minimal length from $y$ to $x$, not going through $z$ (resp. from $x$ to $z$, not going through $y$). 

    Let $X$ be the set of nodes appearing both in $p_{y\rightarrow x}$ and $p_{x\rightarrow z}$. $X$ is not empty, as $x\in X$, and $y,z\notin X$. Let $v\in X$ such that $\dist{\mathcal G}{x}{v}$ is maximal among the nodes of $X$, and let $p_{y\rightarrow v}$ (resp. $p_{v\rightarrow z}$) be the segment of $p_{y\rightarrow x}$ (resp. of $p_{x\rightarrow z}$) from $y$ to $v$ (resp. from $v$ to $z$).

    Then $p_{v\rightarrow z}\cdot (z,y)\cdot p_{y\rightarrow v}$ is a cycle going through an $S$-edge, and is of length $\leq 2r+1$. This is illustrated in Figure~\ref{fig:outside_cones}.

    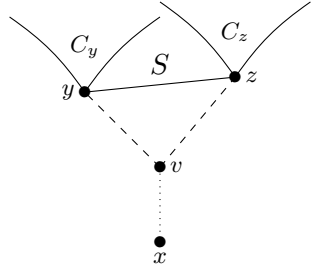
\begin{figure}[ht]
      \centering
      \begin{tikzpicture}[baseline=(current bounding box.center)]
        
        \coordinate (1) at (0,0);
        \coordinate (2) at (-1,2);
        \coordinate (3) at (1,2.2);
        \coordinate (4) at (-2,3);
        \coordinate (5) at (0,3);
        \coordinate (6) at (0,3.2);
        \coordinate (7) at (2,3.2);
        \coordinate (8) at (-1,2.6);
        \coordinate (9) at (1,2.8);
        \coordinate (10) at (0,1);

        \draw (1) node[below] {\footnotesize $x$} node {$\bullet$};
        \draw (2) node[left] {\footnotesize $y$} node {$\bullet$};
        \draw (3) node[right] {\footnotesize $z$} node {$\bullet$};
        \draw (10) node[right] {\footnotesize $v$} node {$\bullet$};

        \draw (8) node {\footnotesize $\cone{y}$};
        \draw (9) node {\footnotesize $\cone{z}$};

        \draw (2) to node[midway,above] {$S$} (3);
        \draw[dashed] (10) to (2);
        \draw[dashed] (10) to (3);
        \draw[dotted] (10) to (1);

        \foreach \from/\to in {2/4,3/6}
        \draw[-][bend right=10] (\from) to (\to);

        \foreach \from/\to in {2/5,3/7}
        \draw[-][bend left=10] (\from) to (\to);
        
      \end{tikzpicture}
      \caption{Existence of a short cycle joining $y$, $z$ and $v$}
      \label{fig:outside_cones}
      
    \end{figure}
    
  \item Otherwise, suppose without loss of generality that $z\in\cone{y}$. This entails that $y\notin\cone{z}$ and $\dist{\mathcal G}{x}{z}=d+1$ where $d:=\dist{\mathcal G}{x}{y}$. 

    Let the initial cycle be $(z,v_1,\cdots,v_{m-1},y)$, with the notation $v_0=z$ and $v_m=y$.

    Let $i$ be the minimal integer such that $v_i\notin\cone{z}$. Let $p_{x\rightarrow v_i}$ be a shortest path from $x$ to $v_i$: by definition, it does not intersect $\cone{z}$, and has length at most $r$. Thus, there exists a path $p_{y\rightarrow v_i}=p_{y\rightarrow x}\cdot p_{x\rightarrow v_i}$ from $y$ to $v_i$ of length at most $r+d$ going only through nodes outside of $\cone{z}$.

    Since $v_{i-1}\in\cone{z}$, there exists a path $p_{v_{i-1}\rightarrow z}$ from $v_{i-1}$ to $z$ of length at most $r-(d+1)$ going only through nodes of $\cone{z}$. 
    
    Hence $p_{y\rightarrow v_i}\cdot(v_i,v_{i-1})\cdot p_{v_{i-1}\rightarrow z}\cdot (z,y)$ is a cycle going though an $S$-edge, and its length is at most $2r+1$. This is depicted in Figure~\ref{fig:inside_cone}. 
    \qedhere

    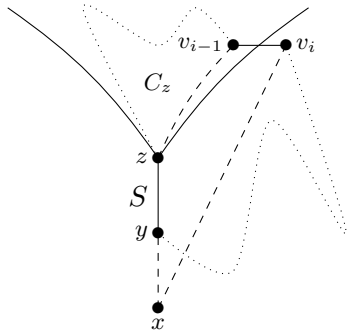
\begin{figure}[ht]
      \centering
      \begin{tikzpicture}[baseline=(current bounding box.center)]

        \coordinate (0) at (0,-1);
        \coordinate (1) at (0,0);
        \coordinate (2) at (0,1);
        \coordinate (3) at (-2,3);
        \coordinate (4) at (2,3);
        \coordinate (5) at (0,2);
        \coordinate (6) at (1,2.5);
        \coordinate (7) at (1.7,2.5);

        \draw (0) node[below] {\footnotesize $x$} node {$\bullet$};
        \draw (1) node[left] {\footnotesize $y$} node {$\bullet$};
        \draw (2) node[left] {\footnotesize $z$} node {$\bullet$};
        \draw (6) node[left] {\footnotesize $v_{i-1}$} node {$\bullet$};
        \draw (7) node[right] {\footnotesize $v_i$} node {$\bullet$};

        \draw (5) node {\footnotesize $\cone{z}$};

        \draw (1) to node[midway,left] {$S$} (2);
        \draw (6) to (7);
        \draw[dashed] (0) to (1);
        \draw[dashed] (0) to (7);
        \draw[dashed][bend left=10] (2) to (6);

        \foreach \from/\to in {2/3}
        \draw[-][bend right=10] (\from) to (\to);

        \foreach \from/\to in {2/4}
        \draw[-][bend left=10] (\from) to (\to);

        \draw [dotted] plot [smooth] coordinates {(1) (1,-0.5) (1.5,1.5) (2.5,0) (7)};
        \draw [dotted] plot [smooth] coordinates {(2) (-1,3) (0,2.5) (0.5,3) (6)};


      \end{tikzpicture}
      \caption{Existence of a short cycle joining $y$, $z$, $v_{i-1}$ and $v_i$}
      \label{fig:inside_cone}
    \end{figure}

  \end{itemize}
\end{proof}

\noindent This characterization of \lay{r} allows us to state the following lemma, which is now straightforward. It provides a method to add an $S$-edge without breaking the property \lay{r}.

\begin{lem}
  \label{lem:layer}
  Let $r\in\N$, and $(\mathcal G,S)$ be a structure satisfying \lay{r}.

  If $x,y\in G$ are such that $\dist{(\mathcal G,S)}{x}{y}> 2r$, then \lay{r} holds in $(\mathcal G,S\cup\{(x,y)\})$
\end{lem}

\section{Proof of Theorem~\ref{th:collapse}}
\label{sec:proof}

We are now ready to prove Theorem~\ref{th:collapse}. Recall the sketch of proof from Section~\ref{sec:prelim}. We proceed in several steps:

Section~\ref{subsec:general_method} details the general framework of the proof. In Section~\ref{subsec:rare_frequent}, we divide the neighborhood types into rare ones and frequent ones.

We then begin the construction of $S_1$: Section~\ref{subsec:rare} is dedicated to the construction of $S_1$ around the occurrences in $\mathcal G_1$ of rare neighborhood types. Then, in Section~\ref{subsec:junction}, we keep constructing $S_1$ around the occurrences (two for each neighborhood type) of frequent neighborhood types that are designed to make, when the construction is complete, the $S_1$-junction between two frequent neighborhood types.

At this point, $S_1$ will be fully built around the singular points of $\mathcal G_1$. Section~\ref{subsec:carry_over} deals with the transfer of this partial successor relation $S_1$ over to $\mathcal G_2$: this will result in a partial $S_2$, built in a similar way around the singular points of $\mathcal G_2$.

In Section~\ref{subsec:completion}, $S_1$ and $S_2$ are completed independently, to cover $G_1$ and $G_2$. These expansions do not need to be coordinated, since at this point, the elements that are not already covered by $S_1$ and $S_2$ are occurrences of frequent neighborhood types and their resulting neighborhood types will be regular (\ie fractal) both in $(\mathcal G_1,S_1)$ and $(\mathcal G_2,S_2)$.

We then give some simple examples in Section~\ref{subsec:examples}, before establishing properties of $S_1$ and $S_2$ in Section~\ref{subsec:properties}, and concluding the proof in Section~\ref{subsec:proof_conclusion}.

\subsection{General method}
\label{subsec:general_method}

Let $\mathcal C$ be a class of $\Sigma$-structures of degree at most $d$. We show the following: for every $\alpha\in\N$, there exists some $f(\alpha)\in\N$ such that, given $\mathcal G_1,\mathcal G_2\in\mathcal C$, if $\mathcal G_1\foeq{f(\alpha)}\mathcal G_2$ then $\mathcal G_1\sieq{\alpha}\mathcal G_2$. For that, we will exhibit successor relations $S_1$ and $S_2$ such that $(\mathcal G_1,S_1)\foeq{\alpha}(\mathcal G_2,S_2)$.

More precisely, using the notations from Definition~\ref{def:types}, we will show that $\threq{(\mathcal G_1,S_1)}{(\mathcal G_2,S_2)}{r}{t}$ where $r$ and $t$ depend on $\alpha$ and are large enough to ensure that $(\mathcal G_1,S_1)\foeq{\alpha}(\mathcal G_2,S_2)$. The existence of such $r$ and $t$ follows from the well-known Hanf threshold theorem, whose finite version is given in \cite{DBLP:journals/iandc/FaginSV95}, and stated as Theorem~4.24 in \cite{DBLP:books/sp/Libkin04}.

We will construct $S_1$ and $S_2$ iteratively in a way that ensures, at each step, that the property \lay{r} holds in $(\mathcal G_1,S_1)$ and in $(\mathcal G_2,S_2)$. The property \lay{r} is obviously satisfied in $(\mathcal G_1,\emptyset)$. Each time we add an $S_1$-edge or an $S_2$-edge, we will make sure that we are in the right conditions to call upon Lemma~\ref{lem:layer}, so that \lay{r} is preserved.

\begin{defi}
  \label{def:size_neigh}
  Let us define the two-variable function $N$ by $N(d,r):=d\cdot\frac{(d-1)^r-1}{d-2}+1$ if $d\neq 2$, and by $N(2,r):=2r+1$.
  
  Note that the size of any $r$-neighborhood of degree at most $d$ is bounded by $N(d,r)$.
\end{defi}

\subsection{Separation between rare and frequent neighborhood types}
\label{subsec:rare_frequent}

Knowing the values of $r$ and $t$ as defined in Section~\ref{subsec:general_method}, we are now able to divide the $r$-neighborhood types into two categories: the \textbf{rare neighborhood types} and the \textbf{frequent neighborhood types}. The intent is that the two structures have the same number of occurrences of every rare neighborhood type, and that frequent neighborhood types have many occurrences (wrt. the total number of occurrences of rare neighborhood types) in both structures. This ``many occurrences wrt.'' is formalized through a function $g$ which is to be specified later on.

More precisely,

\begin{lem}
  \label{lem:rare_freq}
  Given $d,r\in\N$ and an increasing function $g:\N\rightarrow\N$, there exists $p\in\N$ such that for every $\Sigma$-structures $\mathcal G_1,\mathcal G_2\in\mathcal C_d$ satisfying $\mathcal G_1\foeq{p}\mathcal G_2$, we can divide the $r$-neighborhood types over $\Sigma$ of degree at most $d$ into rare neighborhood types and frequent neighborhood types, such that

  \begin{itemize}
  \item every rare neighborhood type has the same number of occurrences in $\mathcal G_1$ and in $\mathcal G_2$

  \item both in $\mathcal G_1$ and in $\mathcal G_2$, every frequent neighborhood type has at least $g(\beta)$ occurrences, where $\beta$ is the number of occurrences of all the rare neighborhood types in the structure
    
  \item if there is no frequent neighborhood type, then $\mathcal G_1$ and $\mathcal G_2$ are isomorphic.
    
  \end{itemize}

\end{lem}

\begin{proof}

  Let $\mathcal G_1$ and $\mathcal G_2$ be such that $\mathcal G_1\foeq{p}\mathcal G_2$, for an integer $p$ whose value will become apparent later in the proof.
  
  Let $\chi_1,\cdots,\chi_n$ be an enumeration of all the $r$-neighborhood types over $\Sigma$ of degree at most $d$, ordered in such a way that $\forall i<j,\nocc{\mathcal G_1}{\chi_i}\leq\nocc{\mathcal G_1}{\chi_j}$. Note that $n$ is a function of $d$ and $r$.

  The classification of neighborhood types between rare ones and frequent ones is done through Algorithm~\ref{alg:types}. The idea is to go through the $r$-neighborhood types in increasing order of occurrences in $\mathcal G_1$; if at some point we reach a neighborhood type with at least $g(\beta)$ occurrences in $\mathcal G_1$, where $\beta$ is the total number of occurrences of the previously visited neighborhood types, then we have found a separation between rare and frequent types. Otherwise, it means that there are few (wrt. $g$) occurrences of each neighborhood type, hence $\mathcal G_1$ is small and, as long as $p$ is large enough, $\mathcal G_2$ must be isomorphic to $\mathcal G_1$. 

  \begin{algorithm}[!ht]
    \small
    
    \caption{Separation between rare and frequent neighborhood types}
    \label{alg:types}
    
    \begin{algorithmic}[1]
      \State $\beta\gets 0$
      \State $i\gets 1$
      \While{$i\leq n$ and $\nocc{\mathcal G_1}{\chi_i}<g(\beta)$}
      \State $\beta\gets\beta+\nocc{\mathcal G_1}{\chi_i}$
      \State $i$++
      \EndWhile 

      \Comment \parbox[t]{.6\linewidth}{If $i\leq n$, $\chi_i$ is the frequent neighborhood type with the least occurrences in $\mathcal G_1$.

      If $i=n+1$, all the neighborhood types are rare.}
    \end{algorithmic}
  \end{algorithm}
  
  At the end of Algorithm~\ref{alg:types}, we call $\chi_1,\cdots,\chi_{i-1}$ the rare neighborhood types, and $\chi_i,\cdots,\chi_n$ the frequent ones.

  Note that $\beta$ indeed counts the total number of occurrences of rare neighborhood types in $\mathcal G_1$.

  We now define the integers $(a_i)_{1\leq i\leq n}$ as $a_1:=g(0)$ and $a_{i+1}:=g(ia_i)$.
  
  As $g$ is monotone, it is easy to show by induction that for each rare neighborhood type~$\chi_j$ with $j<i$, $\nocc{\mathcal G_1}{\chi_j}<a_j$.

  As long as $p$ is chosen large enough so that $\mathcal G_1\foeq{p}\mathcal G_2$ entails $\threq{\mathcal G_1}{\mathcal G_2}{r}{a_n}$, we have by construction that every rare neighborhood type has the same number of occurrences (which is smaller that $a_n$) in $\mathcal G_1$ and in $\mathcal G_2$. Furthermore, in $\mathcal G_1$ as in $\mathcal G_2$, if $\beta$ denotes the total number of occurrences of rare neighborhood types, every frequent neighborhood type has at least $g(\beta)$ occurrences.

  We just need to make sure that the two structures are isomorphic when all the neighborhood types are rare. If this is the case, then $|\mathcal G_1|=|\mathcal G_2|\leq n(a_n-1)$. Hence, as long as $p\geq n(a_n-1)$, $\mathcal G_1\foeq{p}\mathcal G_2$ implies $\mathcal G_1\simeq\mathcal G_2$ when all the neighborhood types are rare.
\end{proof}

Let $\tau_0,\cdots,\tau_{m-1}$ be the frequent neighborhood types. From now on, we suppose that $m\geq 1$: there is nothing to do if $m=0$, since $\mathcal G_1$ and $\mathcal G_2$ are isomorphic. Let $\beta$ be the total number of occurrences of rare neighborhood types in $\mathcal G_1$.

\subsection{Construction of \texorpdfstring{$S_1$}{S1} around elements of rare neighborhood type}
\label{subsec:rare}

To begin with, let us focus on $\mathcal G_1$, and start the construction of $S_1$ around occurrences of rare neighborhood types. Algorithm~\ref{alg:r} deals with this construction. In the following, $R_{\leq k}$ will denote $\bigcup\limits_{0\leq j\leq k} R_j$.

For a given occurrence $x$ of some rare neighborhood type, we choose as its $S_1$-successor and $S_1$-predecessor two occurrences of neighborhood type $\tau_0$ (the first frequent neighborhood type), far apart from one another and from $x$. The existence of those elements relies on the bounded degree hypothesis. This is done on lines~\ref{line:r_succ} and \ref{line:r_pred}.

When line~\ref{line:r_rare_protected} is reached, every occurrence of rare neighborhood type has an $S_1$-predecessor and an $S_1$-successor of neighborhood type $\tau_0$.

\medskip

It is not enough, however, only to deal with the occurrences of rare neighborhood types. We need to ``protect'' them up to distance $r$ in $(\mathcal G_1,S_1)$. For that purpose, we construct the subsets $R_k$ of $G_1$, for $0\leq k\leq r$.

For each $k$, the subset $R_k$ is the set of elements at distance exactly $k$ in $(\mathcal G_1,S_1)$ from the set of occurrences of rare neighborhood types. Until we have reached $k=r$ (that is, distance $r$ from occurrences of rare neighborhood types), every element of $R_k$ is given an $S_1$-successor (line~\ref{line:r_k_succ}) and/or an $S_1$-predecessor (line~\ref{line:r_k_pred}) of its neighborhood type, if it does not already have one. Once again, those elements are required to be far (\ie at distance greater than $2r$) from what already has been constructed.


Provided that $g$ is large enough, it is always possible to find $x^+$ and $x^-$ on lines~\ref{line:r_succ}, \ref{line:r_pred}, \ref{line:r_k_succ} and \ref{line:r_k_pred}. Indeed, all the neighborhood types considered are frequent ones, and the size of the $2r$-neighborhood of $R_{\leq k+1}$ is bounded by a function of $d$, $r$ and $\beta$ (the total number of occurrences of rare neighborhood types in $\mathcal G_1$). More precisely, at any point of the construction, $(\mathcal G_1,S_1)$ has degree at most $d+2$. Hence, the $2r$-neighborhood of $R_r$ has size at most \[\beta N(d+2,3r)\] (recall the definition of $N$ from Definition~\ref{def:size_neigh}), and it is enough to make sure that \[g(\beta)\geq\beta N(d+2,3r)+1\,.\]

\begin{algorithm}[!ht]
  \small

  \caption{Construction of $S_1$ around elements of rare neighborhood type}
  \label{alg:r}

  \begin{algorithmic}[1]
    \State $S_1\gets \emptyset$
    \State $R_0\gets \{x\in G_1:\tp{x}{\mathcal G_1}{r}\text{ is rare}\}$

    \State $R_1,\cdots,R_r\gets\emptyset$

    \ForAll{$x\in R_0$}
    \ForAll{neighbors $y\notin R_{\leq 1}$ of $x$ in $\mathcal G_1$}
    \State $R_1\gets R_1\cup\{y\}$
    \EndFor
    \State{\label{line:r_succ}find \parbox[t]{.5\linewidth}{$x^+$ such that
        
        $\tp{x^+}{\mathcal G_1}{r}=\tau_0$ and
        
        $\dist{(\mathcal G_1,S_1)}{x^+}{R_{\leq 1}}>2r$}}
    
    \Comment{\parbox[t]{.5\linewidth}{We pick a node at distance greater than $2r$ in compliance with Lemma~\ref{lem:layer}, so that neighborhoods stay layered. 
        
        Recall that $\tau_0$ is the first frequent neighborhood type.}}

    \State $R_1\gets R_1\cup\{x^+\}$
    \State $S_1\gets S_1\cup\{(x,x^+)\}$
    \State{\label{line:r_pred}find \parbox[t]{.5\linewidth}{$x^-$ such that

        $\tp{x^-}{\mathcal G_1}{r}=\tau_0$ and
        
        $\dist{(\mathcal G_1,S_1)}{x^-}{R_{\leq 1}}>2r$}}
    \State $R_1\gets R_1\cup\{x^-\}$
    \State $S_1\gets S_1\cup\{(x^-,x)\}$
    \EndFor\label{line:r_rare_protected}

    \Comment{\parbox[t]{.5\linewidth}{At this point, every element of rare neighborhood type has an $S_1$-predecessor and an $S_1$-successor of neighborhood type $\tau_0$}}

    \For{$k$ from $1$ to $r-1$}

    \ForAll{$x\in R_k$}
    \Comment $\tp{x}{\mathcal G_1}{k}$ is a frequent neighborhood type

    \ForAll{neighbors $y\notin R_{\leq k+1}$ of $x$ in $\mathcal G_1$}
    \State $R_{k+1}\gets R_{k+1}\cup\{y\}$
    \EndFor

    \If{$x$ does not have a successor by $S_1$}
    \State{\label{line:r_k_succ}find \parbox[t]{.5\linewidth}{$x^+$ such that

        $\tp{x^+}{\mathcal G_1}{r}=\tp{x}{\mathcal G_1}{r}$ and
        
        $\dist{(\mathcal G_1,S_1)}{x^+}{R_{\leq k+1}}>2r$}}
    \State $R_{k+1}\gets R_{k+1}\cup\{x^+\}$
    \State $S_1\gets S_1\cup\{(x,x^+)\}$
    \EndIf

    \If{$x$ does not have a predecessor by $S_1$}
    \State{\label{line:r_k_pred}find \parbox[t]{.5\linewidth}{$x^-$ such that

        $\tp{x^-}{\mathcal G_1}{r}=\tp{x}{\mathcal G_1}{r}$ and
        
        $\dist{(\mathcal G_1,S_1)}{x^-}{R_{\leq k+1}}>2r$}}
    \State $R_{k+1}\gets R_{k+1}\cup\{x^-\}$
    \State $S_1\gets S_1\cup\{(x^-,x)\}$
    \EndIf

    \EndFor

    \EndFor

  \end{algorithmic}

\end{algorithm}

\subsection{Construction of \texorpdfstring{$S_1$}{S1} around the junctions between two frequent neighborhood types}
\label{subsec:junction}

Recall that there is a second kind of singular elements: those which will be at the junction between two successive frequent neighborhood types. That is, elements of neighborhood type $\tau_i$ that will, in the final structure $(\mathcal G_1,S_1)$, have an $S_1$-successor of neighborhood type $\tau_{i+1[m]}$, or an $S_1$-predecessor of neighborhood type $\tau_{i-1[m]}$.

Those junction elements need to be treated in a similar way as the occurrences of rare neighborhood types in Section~\ref{subsec:rare}. This construction is done following Algorithm~\ref{alg:p}.

\medskip

The idea of Algorithm~\ref{alg:p} is very similar to that of Algorithm~\ref{alg:r}. We start by picking two elements $x_i^\text{min}$ and $x_i^\text{max}$ for every frequent neighborhood type $\tau_i$ (for loop line~\ref{line:p_p0}), that are far from each other and from the previous construction. Once the construction is done, $x_i^\text{min}$ will be the first (wrt. $S_1$) element of the sequence of occurrences of type $\tau_i$, and $x_i^\text{max}$ the last one. The elements that will appear between (in the sense of $S_1$) $x_i^\text{min}$ and $x_i^\text{max}$ will exactly be those of neighrborhood type $\tau_i$, except possibly for $i=0$, where all the occurrences of the rare neighborhood types will also appear between $x_0^\text{min}$ and $x_0^\text{max}$.

Once these $2m$ elements are chosen, we add an $S_1$-edge between each $x_i^\text{max}$ and the corresponding $x_{i+1[m]}^\text{min}$ on line~\ref{line:p_junctions}: in the final structure, these edges will mark the transition (in the sense of $S_1$) between the range of elements of neighborhood type $\tau_i$ and those of neighborhood type $\tau_{i+1[m]}$.

The set $P_0$ of those $2m$ elements will have the same role as the set $R_0$ of occurrences of rare neighborhood types for Algorithm~\ref{alg:r}: we build $S_1$-edges at depth $r$ around it. This is done through the subsets $P_k$ of $G_1$, for $0\leq k\leq r$, $P_k$ being the set of elements at distance $k$ from $P_0$ in $(\mathcal G_1,S_1)$. Once again, $P_{\leq k}$ denotes $\bigcup\limits_{0\leq j\leq k} P_j$.

For the same reason as for Algorithm~\ref{alg:r}, it is always possible to find elements $x^+$ and $x^-$ on lines~\ref{line:p_succ} and \ref{line:p_pred}.

Note that if $m=1$, there is obviously no transition elements: we simply construct an $S_1$-edge between $x_0^\text{max}$ and $x_0^\text{min}$.

\begin{algorithm}[!ht]
  \small
  \caption{Construction of $S_1$ around the junctions between two frequent neighborhood types}
  \label{alg:p}

  \begin{algorithmic}[1]

    \State $P_0,\cdots ,P_r\gets\emptyset$
    
    \For{\label{line:p_p0}$i$ from $0$ to $m-1$}
    \State{find \parbox[t]{.6\linewidth}{$x_i^\text{min}$ such that

        $\tp{x_i^\text{min}}{\mathcal G_1}{r}=\tau_i$ and
        
        $\dist{(\mathcal G_1,S_1)}{x_i^\text{min}}{R_{\leq r}\cup P_0}>2r$}}
    \State $P_0\gets P_0\cup\{x_i^\text{min}\}$
    \State{find \parbox[t]{.6\linewidth}{$x_i^\text{max}$ such that

        $\tp{x_i^\text{max}}{\mathcal G_1}{r}=\tau_i$ and
        
        $\dist{(\mathcal G_1,S_1)}{x_i^\text{max}}{R_{\leq r}\cup P_0}>2r$}}
    \State $P_0\gets P_0\cup\{x_i^\text{max}\}$    
    \EndFor
    
    \For{$i$ from $0$ to $m-1$}
    \State\label{line:p_junctions} $S_1\gets S_1\cup\{(x_i^\text{max},x_{i+1[m]}^\text{min})\}$
    \EndFor

    \For{$k$ from $0$ to $r-1$}
    
    \ForAll{$x\in P_k$}

    \ForAll{neighbors $y\notin P_{\leq k+1}$ of $x$ in $\mathcal G_1$}
    \State $P_{k+1}\gets P_{k+1}\cup\{y\}$
    \EndFor

    \If{$x$ does not have a successor by $S_1$}
    \State{\label{line:p_succ}find \parbox[t]{.6\linewidth}{$x^+$ such that

        $\tp{x^+}{\mathcal G_1}{r}=\tp{x}{\mathcal G_1}{r}$ and 
        
        $\dist{(\mathcal G_1,S_1)}{x^+}{R_{\leq r}\cup P_{\leq k+1}}>2r$}}
    \State $P_{k+1}\gets P_{k+1}\cup\{x^+\}$
    \State $S_1\gets S_1\cup\{(x,x^+)\}$
    \EndIf

    \If{$x$ does not have a predecessor by $S_1$}
    \State{\label{line:p_pred}find \parbox[t]{.6\linewidth}{$x^-$ such that 
        
        $\tp{x^-}{\mathcal G_1}{r}=\tp{x}{\mathcal G_1}{r}$ and
        
        $\dist{(\mathcal G_1,S_1)}{x^-}{R_{\leq r}\cup P_{\leq k+1}}>2r$}}
    \State $P_{k+1}\gets P_{k+1}\cup\{x^-\}$
    \State $S_1\gets S_1\cup\{(x^-,x)\}$
    \EndIf
    \EndFor

    \EndFor
    
  \end{algorithmic}
\end{algorithm}

\subsection{Carrying \texorpdfstring{$S_1$}{S1} over to \texorpdfstring{$\mathcal G_2$}{G2}}
\label{subsec:carry_over}

In Sections~\ref{subsec:rare} and \ref{subsec:junction}, $S_1$ has been constructed around the singular points of $\mathcal G_1$, \ie occurrences of rare neighborhood types and elements that are to make the junction between two $S_1$-segments of frequent neighborhood types. 

Before we extend $S_1$ to the remaining elements (all of them being occurrences of frequent neighborhood types) of $\mathcal G_1$, we carry it over to $\mathcal G_2$. This transfer is possible under the starting hypothesis that $\mathcal G_1$ and $\mathcal G_2$ are \FO-similar.

Indeed, we made the assumption that $\mathcal G_1\foeq{f(\alpha)}\mathcal G_2$; now, provided that $f(\alpha)$ is large enough, this ensures that there exists a substructure of $\mathcal G_2$ which is isomorphic to the part of $\mathcal G_1$ around which we have already constructed $S_1$. This isomorphism will make it possible to carry this partial successor relation over to $\mathcal G_2$. We make this precise in the following.

\medskip

Let \[A_1:=R_{\leq r}\cup P_{\leq r}\] and \[B:=\{x\in G_1:\dist{(\mathcal G_1,S_1)}{x}{A_1}\leq r\}\,.\]

If we let $t_r^d$ be the number of $r$-neighborhood types of degree at most $d$ over $\Sigma$, we must have that $m\leq n$ thus $|A_1|$ can be bounded by \[(\beta+2t_r^d)N(d+2,r)\,.\]

Similarly, the size of $B$ can be bounded by \[(\beta+2t_r^d)N(d+2,2r)\,,\] which is a function of $\beta$, $r$ and $d$. Hence as long as $f(\alpha)$ is larger than that number, the Duplicator has a winning strategy in the Ehrenfeucht-Fra\"iss\'e game between $\mathcal G_1$ and $\mathcal G_2$ in which the Spoiler chooses every element of $B$. Let $h:B\rightarrow G_2$ be the function resulting from such a strategy.


The mapping $h$ defines an isomorphism from $\mathcal G_1|_B$ to $\mathcal G_2|_{\im(h)}$. Let $A_2:=h(A_1)$. By making $f(\alpha)$ large enough, we can make sure that $\im(h)$ covers the $r$-neighborhood in $\mathcal G_2$ of every element of $A_2$. We then have that for every $x\in A_1$, $\tp{h(x)}{\mathcal G_2}{r}=\tp{x}{\mathcal G_1}{r}$.

We set $S_2:=\{(h(x),h(y)):(x,y)\in S_1\}$. $h$ now defines an isomorphism from $(\mathcal G_1,S_1)|_B$ to $(\mathcal G_2,S_2)|_{\im(h)}$, and for every $x\in A_1$, $\tp{h(x)}{(\mathcal G_2,S_2)}{r}=\tp{x}{(\mathcal G_1,S_1)}{r}$.

Note that since $\mathcal G_1$ and $\mathcal G_2$ have the same number of occurrences of each rare neighborhood type, every element lying outside of $A_2$ must have a frequent neighrbohood type.

\subsection{Completion of \texorpdfstring{$S_1$}{S1} and \texorpdfstring{$S_2$}{S2}}
\label{subsec:completion}

Now that $S_1$ and $S_2$ are constructed around all the singular points in $\mathcal G_1$ and $\mathcal G_2$, it remains to extend their construction to all the other elements of the structures. Recall that all the remaining elements are occurrences of frequent neighborhood types.

\medskip

From $(\mathcal G_\epsilon,S_\epsilon)$, for $\epsilon\in\{1,2\}$, at any point in the construction, let us define the partial function $S_\epsilon^*:G_\epsilon\rightarrow G_\epsilon$ that maps $x\in G_\epsilon$ to the (unique) $y$ that is $S_\epsilon$-reachable (while taking the orientation into account) from $x$ and that does not have an $S_\epsilon$-successor. This function is defined on every element that does not belong to an $S_\epsilon$-cycle (and in particular, on every element without an $S_\epsilon$-predecessor).

Likewise, we define $S_\epsilon^{-*}$ by reversing the arrows of $S_\epsilon$.

At this point, for every $x\notin A_1$, $S_1^*(x)=S_1^{-*}(x)=x$, and for every $x\notin A_2$, \[S_2^*(x)=S_2^{-*}(x)=x\,.\]

\medskip

We now run Algorithm~\ref{alg:s}. We first treat $\mathcal G_1$, and then apply a similar method to $\mathcal G_2$, replacing $x_i^\text{min}$ and $x_i^\text{max}$ by $h(x_i^\text{min})$ and $h(x_i^\text{max})$. The idea is, for every frequent neighborhood type $\tau_i$, to insert all its remaining occurrences between (in the sense of $S_1$) $x_i^\text{min}$ and $x_i^\text{max}$.

The first approach (the loop at line~\ref{line:greedy}) is greedy: while constructing $S_\epsilon$ on nodes of neighborhood type $\tau_i$, we choose as the successor of the current node any occurrence of $\tau_i$ that is at distance greater than $2r$ from the current node $s$ and the closing node of neighborhood type $\tau_i$, $S_\epsilon^{-*}(x_i^\text{max})$. This, together with Lemma~\ref{lem:layer}, ensures that \lay{r} holds after every addition. The conditions line~\ref{line:greedy_cond} also ensure that the final edge addition, line~\ref{line:close}, does not break \lay{r}.

Once we cannot apply this greedy approach anymore, we know that only a small number (which can be bounded by $2N(d+2,2r)$) of nodes of neighborhood type $\tau_i$ remain without $S_1$-predecessor. The loop at line~\ref{line:cycle} considers one such node $x$ at a time. As long as $g$ is large enough, we have constructed $S_1$ around enough elements of neighborhood type $\tau_i$ in the greedy approach to ensure the existence of some $S_1(y,z)$, with $y,z$ of neighborhood type $\tau_i$ and at distance greater than $2r$ from $x$; $x$ is inserted between $y$ and $z$ (line~\ref{line:cut}). For that, it is enough to have constructed at least \[2N(d+2,2r)+1\] $S_1$-edges in the greedy phase. This is the case in particular when there are at least \[4N(d+2,2r)+1\] elements of neighborhood type $\tau_i$ without $S_1$-predecessor at the beginning of Algorithm~\ref{alg:s}, which can be ensured by having \[g(\beta)\geq |A_1|+4N(d+2,2r)+1\,.\] This holds in particular when \[g(\beta)\geq (\beta+2t_r^d)N(d+2,r)+4N(d+2,2r)+1\,.\]

We will prove in Lemma~\ref{lem:lay} that all these insertions preserve \lay{r}.



\begin{algorithm}[!ht]
  \small
  \caption{Completion of $S_\epsilon$}
  \label{alg:s}
  
  \begin{algorithmic}[1]

    \For{$\epsilon$ from $1$ to $2$}

    \For{$i$ from $0$ to $m-1$} \label{line:greedy}

    \If{$\epsilon=1$}
    \State $s\gets S_1^*(x_i^\text{min})$ \label{line:new_s}
    \State $t\gets S_1^{-*}(x_i^\text{max})$
    \Else
    \State $s\gets S_2^*(h(x_i^\text{min}))$
    \State $t\gets S_2^{-*}(h(x_i^\text{max}))$
    \EndIf

    \While{such an $x$ exists} \label{line:greedy_while}
    \State{find \parbox[t]{.6\linewidth}{$x$ with no $S_\epsilon$-predecessor, such that \label{line:greedy_cond}
        $\tp{x}{\mathcal G_\epsilon}{r}=\tau_i$, 
  
        $\dist{(\mathcal G_\epsilon,S_\epsilon)}{s}{x}>2r$, 

        $\dist{(\mathcal G_\epsilon,S_\epsilon)}{x}{t}>2r$ and
        
        $\dist{(\mathcal G_\epsilon,S_\epsilon)}{S_\epsilon^*(x)}{t}>2r$}}
    \State $S_\epsilon\gets S_\epsilon\cup\{(s,x)\}$ \label{line:next}
    \State $s\gets S_\epsilon^*(x)$ \label{line:s<-x}
    \EndWhile

    \Comment{\parbox[t]{.5\linewidth}{At this point, only a bounded number of elements of neighborhood type $\tau_i$ are left without an $S_\epsilon$-predecessor}}
    
    \State $S_\epsilon\gets S_\epsilon\cup\{(s,t)\}$ \label{line:close}

    \EndFor
    
    \For{$i$ from $0$ to $m-1$}\label{line:cycle}

    \ForAll{$x$ without $S_\epsilon$-predecessor, s.t. $\tp{x}{\mathcal G_\epsilon}{r}=\tau_i$}
    \State{find \parbox[t]{.5\linewidth}{ $y,z\notin A_\epsilon$ such that

        $\tp{y}{\mathcal G_\epsilon}{r}=\tp{z}{\mathcal G_\epsilon}{r}=\tau_i$, 

        $(y,z)\in S_\epsilon$, 

        $\dist{(\mathcal G_\epsilon,S_\epsilon)}{y}{x}>2r$ and

        $\dist{(\mathcal G_\epsilon,S_\epsilon)}{S_\epsilon^*(x)}{z}>2r$}}
    \State $S_\epsilon\gets S_\epsilon\setminus\{(y,z)\}\cup\{(y,x),(S_\epsilon^*(x),z)\}$ \label{line:cut}
    \EndFor
    
    \EndFor    

    \EndFor
  \end{algorithmic}
\end{algorithm}

\subsection{Examples of construction}
\label{subsec:examples}

Before we give the proof of correctness of these algorithms, let us see how they apply in some simple cases

\begin{exa}
  Suppose that there are no occurrences of rare neighborhood types, and only one frequent neighborhood type $\tau_0$, and assume $r=2$.

  In this case, Algorithm~\ref{alg:r} is irrelevant, and all Algorithm~\ref{alg:p} does is pick $x_0^\text{max}$ and $x_0^\text{min}$ far from each other, and start building $S_1$ around those nodes in order to construct their complete $r$-neighborhood in $(\mathcal G_1,S_1)$. In order to make the figure more readable, let us consider that $x_0^\text{max}$ and $x_1^\text{min}$ have only one neighbor. In Figure~\ref{fig:p_1}, the plain lines represent edges in $\mathcal G_1$, and the dashed arrows represent $S_1$.

  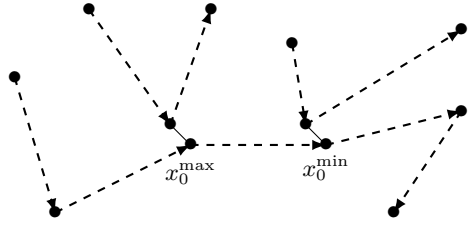
\begin{figure}[ht]
    \centering
    \begin{tikzpicture}[baseline=(current bounding box.center),scale=0.9]
        
      \coordinate (1) at (0,0);
      \coordinate (2) at (2,0);
      \coordinate (3) at (-0.3,0.3);
      \coordinate (4) at (1.7,0.3);

      \coordinate (5) at (-1.5,2);
      \coordinate (6) at (-2,-1);

      \coordinate (7) at (-2.6,1);
      \coordinate (8) at (0.3,2);

      \coordinate (9) at (4,0.5);
      \coordinate (10) at (3,-1);

      \coordinate (11) at (1.5,1.5);
      \coordinate (12) at (4,1.7);

      \draw (1) node[below=3pt] {\footnotesize $x_0^\text{max}$} node {$\bullet$};
      \draw (2) node[below] {\footnotesize $x_0^\text{min}$} node {$\bullet$};

      \foreach \node in {3,...,12}
      \draw (\node) node {$\bullet$};

      \foreach \from/\to in {1/3,2/4}
      \draw[-] (\from) to (\to);

      \foreach \from/\to in {1/2,7/6,6/1,5/3,3/8,2/9,9/10,11/4,4/12}
      \draw[->,thick,dashed,>=latex] (\from) to (\to);

    \end{tikzpicture}
    \caption{After Algorithm~\ref{alg:p}, with one frequent neighborhood type}
    \label{fig:p_1}
  \end{figure}

  We now apply Algorithm~\ref{alg:s}. The first step is to add elements between (in the sense of $S_1$) $S_1^*(x_0^\text{min})$ and $S_1^{-*}(x_0^\text{max})$ in order to join them, in a greedy fashion. Once this is done, there only remain a few elements that have not been assigned an $S_1$-predecessor. This is depicted in Figure~\ref{fig:s_1}.

  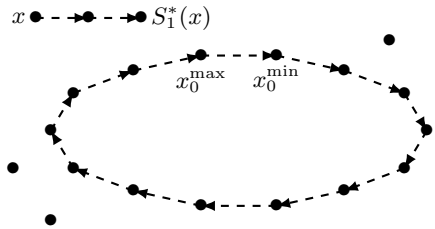
\begin{figure}[ht]
    \centering
    \begin{tikzpicture}[baseline=(current bounding box.center),scale=1]
      
      \coordinate (1) at (-0.5,1);
      \coordinate (2) at (0.5,1);
      \coordinate (3) at (1.4,0.8);
      \coordinate (4) at (2.2,0.5);
      \coordinate (5) at (2.5,0);
      \coordinate (6) at (2.2,-0.5);
      \coordinate (7) at (1.4,-0.8);
      \coordinate (8) at (0.5,-1);
      \coordinate (9) at (-0.5,-1);
      \coordinate (10) at (-1.4,-0.8);
      \coordinate (11) at (-2.2,-0.5);
      \coordinate (12) at (-2.5,0);
      \coordinate (13) at (-2.2,0.5);
      \coordinate (14) at (-1.4,0.8);

      \coordinate (15) at (2,1.2);

      \coordinate (16) at (-2.7,1.5);
      \coordinate (17) at (-2,1.5);
      \coordinate (18) at (-1.3,1.5);

      \coordinate (19) at (-2.5,-1.2);
      \coordinate (20) at (-3,-0.5);

      \draw (1) node[below=2pt] {\footnotesize $x_0^\text{max}$};
      \draw (2) node[below] {\footnotesize $x_0^\text{min}$};
      \draw (16) node[left] {\footnotesize $x$};
      \draw (18) node[right] {\footnotesize $S_1^*(x)$};

      \foreach \node in {1,...,20}
      \draw (\node) node {$\bullet$};

      \foreach \from/\to in {1/2,2/3,3/4,4/5,5/6,6/7,7/8,8/9,9/10,10/11,11/12,12/13,13/14,14/1,16/17,17/18}
      \draw[->,thick,dashed,>=latex] (\from) to (\to);

    \end{tikzpicture}
    \caption{After the greedy part of Algorithm~\ref{alg:s}, with one frequent neighborhood type}
    \label{fig:s_1}
  \end{figure}

  Now we consider one by one each of the elements that do not have an $S_1$-predecessor: let us start with $x$ in Figure~\ref{fig:s_1}. Our goal is to insert it in the $S_1$-cycle while still respecting \lay{r}. For that, we find two successive elements $y,z$ of the cycle that are far from $x$ and $S_1^*(x)$, and we insert $x$ between them, as shown in Figure~\ref{fig:s_2}.

  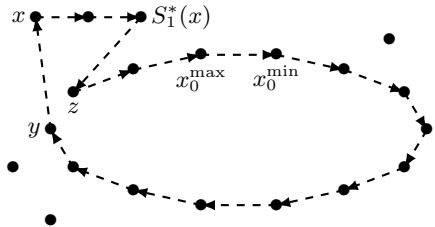
\begin{figure}[ht]
    \centering
    \begin{tikzpicture}[baseline=(current bounding box.center),scale=1]
      
      \coordinate (1) at (-0.5,1);
      \coordinate (2) at (0.5,1);
      \coordinate (3) at (1.4,0.8);
      \coordinate (4) at (2.2,0.5);
      \coordinate (5) at (2.5,0);
      \coordinate (6) at (2.2,-0.5);
      \coordinate (7) at (1.4,-0.8);
      \coordinate (8) at (0.5,-1);
      \coordinate (9) at (-0.5,-1);
      \coordinate (10) at (-1.4,-0.8);
      \coordinate (11) at (-2.2,-0.5);
      \coordinate (12) at (-2.5,0);
      \coordinate (13) at (-2.2,0.5);
      \coordinate (14) at (-1.4,0.8);

      \coordinate (15) at (2,1.2);

      \coordinate (16) at (-2.7,1.5);
      \coordinate (17) at (-2,1.5);
      \coordinate (18) at (-1.3,1.5);

      \coordinate (19) at (-2.5,-1.2);
      \coordinate (20) at (-3,-0.5);

      \draw (1) node[below=2pt] {\footnotesize $x_0^\text{max}$};
      \draw (2) node[below] {\footnotesize $x_0^\text{min}$};
      \draw (12) node[left] {\footnotesize $y$};
      \draw (13) node[below] {\footnotesize $z$};
      \draw (16) node[left] {\footnotesize $x$};
      \draw (18) node[right] {\footnotesize $S_1^*(x)$};

      \foreach \node in {1,...,20}
      \draw (\node) node {$\bullet$};

      \foreach \from/\to in {1/2,2/3,3/4,4/5,5/6,6/7,7/8,8/9,9/10,10/11,11/12,13/14,14/1,16/17,17/18,12/16,18/13}
      \draw[->,thick,dashed,>=latex] (\from) to (\to);

    \end{tikzpicture}
    \caption{Inserting $x$ in the $S_1$-cycle, as in the second part of Algorithm~\ref{alg:s}, with one frequent neighborhood type}
    \label{fig:s_2}
  \end{figure}

  We treat all the elements without an $S_1$-predecessor in the same way, until $S_1$ is fully built.

\end{exa}

\begin{exa}
  Suppose now that there are two frequent neighborhood types $\tau_0$ and $\tau_1$, and still no occurrences of rare neighborhood types.
  
  The procedure is very similar: in Algorithm~\ref{alg:p}, we build the $r$-neighborhood in $(\mathcal G_1,S_1)$ of the four nodes $x_0^\text{max}$, $x_0^\text{min}$, $x_1^\text{max}$ and $x_1^\text{min}$. 

  After the greedy part of Algorithm~\ref{alg:s}, $S_1$ looks like in Figure~\ref{fig:s_3}, where occurrences of $\tau_0$ are represented as $\bullet$ and occurrences of $\tau_1$ as $\circ$. The remaining of Algorithm~\ref{alg:s} is as unchanged.
  
  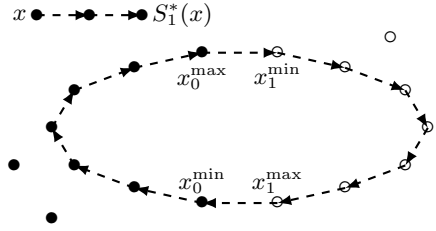
\begin{figure}[ht]
    \centering 
    \begin{tikzpicture}[baseline=(current bounding box.center)]
      
      \coordinate (1) at (-0.5,1);
      \coordinate (2) at (0.5,1);
      \coordinate (3) at (1.4,0.8);
      \coordinate (4) at (2.2,0.5);
      \coordinate (5) at (2.5,0);
      \coordinate (6) at (2.2,-0.5);
      \coordinate (7) at (1.4,-0.8);
      \coordinate (8) at (0.5,-1);
      \coordinate (9) at (-0.5,-1);
      \coordinate (10) at (-1.4,-0.8);
      \coordinate (11) at (-2.2,-0.5);
      \coordinate (12) at (-2.5,0);
      \coordinate (13) at (-2.2,0.5);
      \coordinate (14) at (-1.4,0.8);

      \coordinate (15) at (-3,-0.5);

      \coordinate (16) at (-2.7,1.5);
      \coordinate (17) at (-2,1.5);
      \coordinate (18) at (-1.3,1.5);

      \coordinate (19) at (-2.5,-1.2);
      \coordinate (20) at (2,1.2);

      \draw (1) node[below=2pt] {\footnotesize $x_0^\text{max}$};
      \draw (2) node[below] {\footnotesize $x_1^\text{min}$};
      \draw (8) node[above] {\footnotesize $x_1^\text{max}$};
      \draw (9) node[above] {\footnotesize $x_0^\text{min}$};
      \draw (16) node[left] {\footnotesize $x$};
      \draw (18) node[right] {\footnotesize $S_1^*(x)$};

      \foreach \node in {2,...,8,20}
      \draw (\node) node {$\circ$};

      \foreach \node in {9,...,19,1}
      \draw (\node) node {$\bullet$};

      \foreach \from/\to in {1/2,2/3,3/4,4/5,5/6,6/7,7/8,8/9,9/10,10/11,11/12,12/13,13/14,14/1,16/17,17/18}
      \draw[->,thick,dashed,>=latex] (\from) to (\to);

    \end{tikzpicture}
    \caption{After the greedy part of Algorithm~\ref{alg:s}, with two frequent neighborhood types}
    \label{fig:s_3}
  \end{figure}

  Note that if there existed some occurrences of rare neighborhood types, they would be embedded in the $\tau_0$ part of the $S_1$-cycle.

\end{exa}

\subsection{Properties of \texorpdfstring{$S_1$}{S1} and \texorpdfstring{$S_2$}{S2}}
\label{subsec:properties}

We are now ready to show that, after the successive run of Algorithms~\ref{alg:r}, \ref{alg:p} and \ref{alg:s},
\begin{itemize}
\item $S_1$ and $S_2$ are indeed successor relations (Lemma~\ref{lem:succ}),
\item $(\mathcal G_1,S_1)$ and $(\mathcal G_2,S_2)$ satisfy \lay{r} (Lemma~\ref{lem:lay}), and
\item a singular element (around a rare or a junction element) of $(\mathcal G_1,S_1)$ and its corresponding element via $h$ in $(\mathcal G_2,S_2)$ have the same $r$-neighborhood type (Lemma~\ref{lem:h_preserves_types}), while any other element in both structures has a regular (\ie fractal) $r$-neighborhood type (Lemma~\ref{lem:enriched_types}).
\end{itemize}
These properties will allow us to prove in Section~\ref{subsec:proof_conclusion} that $(\mathcal G_1,S_1)$ and $(\mathcal G_2,S_2)$ have the same number of occurrences of every $r$-neighborhood type, up to a threshold $t$.

\begin{lem}
  \label{lem:succ}
  $S_1$ (resp. $S_2$) is a successor relation on $G_1$ (resp. $G_2$).
\end{lem}

\begin{proof}
  This result is rather transparent, but a rigorous proof requires the usage of a somewhat cumbersome invariant.

  Let us focus on $\mathcal G_1$: the proof is the same for $\mathcal G_2$, replacing every $x_i^\text{min}$ and $x_i^\text{max}$ with $h(x_i^\text{min})$ and $h(x_i^\text{max})$.

  \medskip
  
  Let $a\in G_1$ be defined as $S_1^{-*}(x_{m-1}^\text{max})$ at the beginning of Algorithm~\ref{alg:s}. By construction, $\tp{a}{\mathcal G}{r}=\tau_{m-1}$ and $a$ has no $S_1$-predecessor as of now.

  We show that at any point before line~\ref{line:close} of the loop iteration $i=m-1$ of Algorithm~\ref{alg:s}, 

  \begin{enumerate}[(i)]
  \item\label{s} $S_1^{-*}(s)=a$
  \item\label{a} $S_1^{-*}(x_{m-1}^\text{max})=a$
  \item\label{yz} let $y,z\notin A_1$ such that $(y,z)\in S_1$ and \[\tp{y}{\mathcal G_1}{r}=\tp{z}{\mathcal G_1}{r}=\tau_j\] for some $j$; then $S_1^{-*}(y)=a$
  \item\label{p} for every $i$, $(x_i^\text{max},x_{i+1[m]}^\text{min})\in S_1$
  \item\label{cycle} there is no $S_1$-cycle
  \item\label{j} for every $j>i$, $\tp{S_1^{-*}(x_j^\text{max})}{\mathcal G_1}{r}=\tau_j$
  \end{enumerate}
  This is obviously satisfied at the beginning of Algorithm~\ref{alg:s}: there are not yet such $y,z$ as in~\ref{yz}, and $s=S_1^*(x_0^\text{min})$ is $S_1$-reachable from $x_{m-1}^\text{max}$ (since $(x_{m+1}^\text{max},x_0^\text{min})\in S_1$) hence \ref{s} holds.

    Line~\ref{line:new_s} preserves the invariant. Indeed, the new value of $s$ is $S_1$-reachable from its previous value (this is guaranteed by \ref{p}), which means that they have the same image through $S_1^{-*}$, namely $a$.

  Let us prove that line~\ref{line:next} preserves the invariant. \ref{s} and \ref{a} still hold since $x\neq a$: indeed, for $i<m-1$, $x$ and $a$ do not share the same neighborhood type, while for $i=m-1$, $a=t$ (because of \ref{a}) and the distance condition prohibits $x=a$. \ref{yz} still holds, as the only new possibility for such a couple $(y,z)$ is $(s,x)$, which is such that $S_1^{-*}(y)=a$ (because of \ref{s}). \ref{p} obviously holds, as does \ref{cycle}, since the only way for an $S_1$-cycle to have been created is if $x=S_1^{-*}(s)$, that is $x=a$. We have seen that this is absurd. \ref{j} is satisfied, as the only way for it to fail is for $x$ to be some $S_1^{-*}(x_j^\text{max})$, for $j>i$, which is impossible due to neighborhood type requirements.

  Now, let us move to line~\ref{line:s<-x}. Only \ref{s} needs verification, and the argument is the same as for line~\ref{line:new_s}.

  Finally, let us look at line~\ref{line:close}, for $i<m-1$. $t\neq a$ since their neighborhood types are different, hence \ref{s}, \ref{a} and \ref{cycle} still hold. \ref{yz} still holds, as the only new possibility for such a couple $(y,z)$ is $(s,t)$, which is such that $S_1^{-*}(y)=a$ because of \ref{s} (actually, $(s,t)$ does not even fit the condition, since $t\in A_1$). \ref{p} is still satisfied. \ref{j} holds, as the only way for it to fail is for $t$ to be some $S_1^{-*}(x_j^\text{max})$, for $j>i$, which is impossible due to neighborhood type requirements.

  \medskip

    We now prove that from line~\ref{line:cycle} until the end of Algorithm~\ref{alg:s}, there is exactly one $S_1$-cycle, which contains every $y,z\notin A_1$ such that $(y,z)\in S_1$ and $\tp{y}{\mathcal G_1}{r}=\tp{z}{\mathcal G_1}{r}=\tau_j$ for some $j$.

  This is true after line~\ref{line:close} of the loop iteration $m-1$, which creates the first $S_1$-cycle, as \ref{s} and \ref{a} ensure $t=a=S_1^{-*}(s)$. \ref{yz} guarantees that this newly created $S_1$-cycle contains all the couple $(y,z)$ satisfying the condition.
  
  It remains to show that line~\ref{line:cut} preserves this property: by hypothesis, $y$ and $z$ belong to the $S_1$-cycle. After line~\ref{line:cut}, there is still exactly one $S_1$-cycle, which corresponds to the previous one where the $S_1$-edge has been replaced by the $S_1$-segment $[x,S_1^*(x)]$. The only $S_1$-edges that have been added belong to the $S_1$-cycle, hence the second part of the property still holds.  
  
  \medskip
  
  In the end, every element of $G_1$ has a predecessor by $S_1$, hence $S_1$ is a permutation of~$G_1$. We have shown that it has a single orbit.
\end{proof}

\begin{lem}
  \label{lem:lay}
  \lay{r} holds in $(\mathcal G_\epsilon,S_\epsilon)$, for $\epsilon\in\{1,2\}$
\end{lem}

\begin{proof}
  This property is guaranteed by the distance conditions of the form \[\dist{(\mathcal G_\epsilon,S_\epsilon)}{.}{.}>2r\] imposed throughout Algorithms~\ref{alg:r}, \ref{alg:p} and \ref{alg:s}, and by Lemma~\ref{lem:layer}.

  \medskip
  
  One can very easily verify that \lay{r} is guaranteed by Lemma~\ref{lem:layer} to hold in $(\mathcal G_\epsilon,S_\epsilon)$ prior to the run of Algorithm~\ref{alg:s}.

  We focus on Algorithm~\ref{alg:s}, and we use Lemma~\ref{lem:layer} to prove that \lay{r} remains valid in $(\mathcal G_\epsilon,S_\epsilon)$ throughout its run. There are three edge additions we have to prove correct:

  \begin{itemize}
  \item For the edge addition of line~\ref{line:next}, this follows directly from Lemma~\ref{lem:layer}.

  \item For the edge addition of line~\ref{line:close}, we show that the invariant \[\dist{(\mathcal G_\epsilon,S_\epsilon)}{s}{t}>2r\] is satisfied at the beginning and at the end of the while at line~\ref{line:greedy_while}. This invariant, together with Lemma~\ref{lem:layer}, will be enough to conclude.

    The invariant holds before the first execution of the while loop.

    Working towards a contradiction, suppose that the invariant is broken during an execution of the loop. In the following, when mentioning a variable, we refer to its value at the beginning of the loop. There must exists a path from $S_\epsilon^*(x)$ (which is to become the new value of $s$ at the end of the loop) to $t$ in $(\mathcal G_\epsilon,S_\epsilon\cup\{(s,x)\})$ of length at most $2r$; consider a shortest one. As is cannot be valid in $(\mathcal G_\epsilon,S_\epsilon)$ by choice of $x$, it must go through the newly added edge $(s,x)$. This means that in $(\mathcal G_\epsilon,S_\epsilon)$, either there exist paths of length at most $2r$ from $S_\epsilon^*(x)$ to $s$ and from $x$ to $t$, or paths of length at most $2r$ from $S_\epsilon^*(x)$ to $x$ and from $s$ to $t$. The former is absurd considering the way $x$ was chosen, and the latter contradicts the previous invariant.

  \item Let us prove that the addition of the two $S_\epsilon$-edges of line~\ref{line:cut} does not break \lay{r}. By choice of $y$, we know that $\dist{(\mathcal G_\epsilon,S_\epsilon)}{y}{x}>2r$. \emph{A fortiori}, we must have \[\dist{(\mathcal G_\epsilon,S_\epsilon\setminus\{(y,z)\})}{y}{x}>2r\,,\] and Lemma~\ref{lem:layer} ensures that $(\mathcal G_\epsilon,S_\epsilon\setminus\{(y,z)\}\cup\{(y,x)\})$ satisfies the property \lay{r}.
    
    Now, to the second addition: let us prove that, at the beginning of line~\ref{line:cut}, \[\dist{(\mathcal G_\epsilon,S_\epsilon\setminus\{(y,z)\}\cup\{(y,x)\})}{S_\epsilon^*(x)}{z}>2r\,.\] We then conclude with Lemma~\ref{lem:layer}.
    
    Suppose it is not the case and consider a shortest path from $S_\epsilon^*(x)$ to $z$, which must be of length at most $2r$. This path cannot be valid in $(\mathcal G_\epsilon,S_\epsilon)$, thus it has to go through the new edge $(y,x)$. Since there cannot exist a path of length at most $2r$ from $S_\epsilon^*(x)$ to $y$ in $(\mathcal G_\epsilon,S_\epsilon)$ (as this would contradict $\dist{(\mathcal G_\epsilon,S_\epsilon)}{S_\epsilon^*(x)}{z}>2r$), it has to borrow the edge from $x$ to $y$.

    Then in $(\mathcal G_\epsilon,S_\epsilon\setminus\{(y,z)\})$, there is a path of length at most $2r$ from $y$ to $z$, which contradicts \lay{r} in $(\mathcal G_\epsilon,S_\epsilon)$.
    \qedhere
  \end{itemize}
\end{proof}
\noindent The following Lemma states that the only time $S_\epsilon$ joins two nodes that have different $r$-neighborhood types in $\mathcal G_\epsilon$ is when one of them is an occurrence of a rare neighborhood type (in which case its $S_\epsilon$-predecessor and $S_\epsilon$-successor are of neighborhood type $\tau_0$) or when they are the elements which make the transition between two frequent neighborhood types (that is, one is $x_i^\text{max}$ and the other is $x_{i+1[m]}^\text{min}$, for some $i<m$):

\begin{lem}
  \label{lem:succ_types}We have the following:

  \begin{itemize}
  \item $\forall x,y\in G_1$ such that $(x,y)\in S_1$ and ($x\notin R_0$ and $y\notin R_0$) and ($x\notin P_0$ or $y\notin P_0$), then $\tp{x}{\mathcal G_1}{r}=\tp{y}{\mathcal G_1}{r}$

  \item  $\forall x,y\in G_2$ such that $(x,y)\in S_2$ and ($x\notin h(R_0)$ and $y\notin h(R_0)$) and ($x\notin h(P_0)$ or $y\notin h(P_0)$), then $\tp{x}{\mathcal G_2}{r}=\tp{y}{\mathcal G_2}{r}$
  \end{itemize}
\end{lem}

\begin{proof}
  The property clearly holds at the end of Algorithm~\ref{alg:r} and Algorithm~\ref{alg:p}.

  For any $i$ from $0$ to $m-1$, the only $S_1$-edges (resp. $S_2$-edges) that are added during the $i$-th loop are between two nodes of neighborhood type $\tau_i$.
\end{proof}

Recall the discussion at the beginning of Section~\ref{sec:fractal}. We now prove that, as long as an element is far from any occurrence of a rare neighborhood type and from the elements that make the transition between two frequent neighborhood types, its neighborhood type in $(\mathcal G_\epsilon,S_\epsilon)$ is the fractal of its neighborhood type in $\mathcal G_\epsilon$:

\begin{lem}
  \label{lem:enriched_types}
  For $\epsilon\in\{1,2\}$ and for every $0\leq k\leq r$ and $x\notin R_{\leq k}\cup P_{\leq k}$ (if $\epsilon=1$) or $x\notin h(R_{\leq k}\cup P_{\leq k})$ (if $\epsilon=2$), 
  \[\tp{x}{(\mathcal G_\epsilon,S_\epsilon)}{k}=\fractype{\tp{x}{\mathcal G_\epsilon}{k}}{k}\,.\]
\end{lem}

\begin{proof}
  
  We prove the result by induction on $k$. For $k=0$, there is nothing to do but note that no edge $S_\epsilon(x,x)$ has been created.

  Suppose that we have proven the result for some $k<r$, and let $x\notin R_{\leq k+1}\cup P_{\leq k+1}$, or $x\notin h(R_{\leq k+1}\cup P_{\leq k+1})$.
  
  Let $y$ be such that $\dist{\mathcal G_\epsilon}{x}{y}=d$, for some $1\leq d\leq k+1$. By construction of the $R_i$ and $P_i$, and of $h$, we have that $y\notin R_{\leq k+1-d}\cup P_{\leq k+1-d}$, or $y\notin h(R_{\leq k+1-d}\cup P_{\leq k+1-d})$ (this is easily shown by induction on $d$). Hence, $\tp{y}{(\mathcal G_\epsilon,S_\epsilon)}{k+1-d}=\fractype{\tp{y}{\mathcal G_\epsilon}{k+1-d}}{k+1-d}$.

  Because Lemma~\ref{lem:lay} ensures that the $(k+1)$-neighborhood of $x$ in $(\mathcal G_\epsilon,S_\epsilon)$ is layered, it only remains to show that the $S_\epsilon$-successor $x^+$ and predecessor $x^-$ of $x$ are such that $\tp{x^+}{(\mathcal G_\epsilon,S_\epsilon)}{k}=\tp{x^-}{(\mathcal G_\epsilon,S_\epsilon)}{k}=\fractype{\tp{x}{\mathcal G_\epsilon}{k}}{k}$. Let us show this for $x^+$.
  
  Lemma~\ref{lem:succ_types} ensures $\tp{x^+}{\mathcal G_\epsilon}{r}=\tp{x}{\mathcal G_\epsilon}{r}$. It only remains to note that $x^+\notin R_{\leq k}\cup P_{\leq k}$, or $x^+\notin h(R_{\leq k}\cup P_{\leq k})$, and the induction hypothesis allows us to conclude.
\end{proof}

When we first defined $h$, it preserved $r$-neighborhood types by construction. The last step before we are able to conclude the proof of Theorem~\ref{th:collapse} is to make sure that $h$ still preserves $r$-neighborhood types, taking into account the $S_\epsilon$-edges added during the run of Algorithm~\ref{alg:s}:

\begin{lem}
  \label{lem:h_preserves_types}
  $\forall x\in A_1, \tp{h(x)}{(\mathcal G_2,S_2)}{r}=\tp{x}{(\mathcal G_1,S_1)}{r}$
\end{lem}

\begin{proof}
  We prove by induction on $0\leq k\leq r$ that $\forall x\in A_1, \tp{h(x)}{(\mathcal G_2,S_2)}{k}=\tp{x}{(\mathcal G_1,S_1)}{k}$

  There is nothing to prove for $k=0$.

  Moving from $k$ to $k+1$, let $x\in A_1$ and let $y$ be such that $\dist{\mathcal G_1}{x}{y}=d$, for some $1\leq d\leq k+1$. Note that $y\in B$, hence it has an image by $h$.

  If $y\in A_1$, the induction hypothesis allows us to conclude that \[\tp{h(y)}{(\mathcal G_2,S_2)}{k+1-d}=\tp{y}{(\mathcal G_1,S_1)}{k+1-d}\,.\]

  Else, Lemma~\ref{lem:enriched_types} ensures that: \[\tp{h(y)}{(\mathcal G_2,S_2)}{r}=\fractype{\tp{h(y)}{\mathcal G_2}{r}}{r}=\fractype{\tp{y}{\mathcal G_1}{r}}{r}=\tp{y}{(\mathcal G_1,S_1)}{r}\,.\]
  In both cases, $\tp{h(y)}{(\mathcal G_2,S_2)}{k+1-d}=\tp{y}{(\mathcal G_1,S_1)}{k+1-d}$.

  Because of \lay{r}, it only remains to show that the $S_\epsilon$-successors of $x$ and $h(x)$, as well as their $S_\epsilon$-predecessors, have the same $k$-neighborhood type in $(\mathcal G_\epsilon,S_\epsilon)$. Let us prove this for the successors, respectively named $x^+$ and $h(x)^+$.

  If $x^+\in A_1$, then by construction $h(x)^+=h(x^+)$, and the induction hypothesis allows us to conclude.
  
  Otherwise, $x^+\notin A_1$ and $h(x)^+\notin A_2$. Under this hypothesis, Lemma~\ref{lem:succ_types} ensures that \[\tp{h(x)^+}{\mathcal G_2}{r}=\tp{h(x)}{\mathcal G_2}{r}=\tp{x}{\mathcal G_1}{r}=\tp{x^+}{\mathcal G_1}{r}\,.\] 
  Now, Lemma~\ref{lem:enriched_types} ensures that

  {\renewcommand{\arraystretch}{1.4}
    \begin{tabular}{rcl}
      $\tp{h(x)^+}{(\mathcal G_2,S_2)}{r}$&$=$&$\fractype{\tp{h(x)^+}{\mathcal G_2}{r}}{r}$\\
                                          &$=$&$\fractype{\tp{x^+}{\mathcal G_1}{r}}{r}$\\
                                          &$=$&$\tp{x^+}{(\mathcal G_1,S_1)}{r}\,.$\\
    \end{tabular}
  }
  
  \qedhere
\end{proof}

\subsection{Conclusion of the proof}
\label{subsec:proof_conclusion}

We recall Theorem~\ref{th:collapse}, whose proof we are now able to conclude:

\collapse* 

Let $\alpha\in\N$. We want to prove that there exists some $f(\alpha)\in\N$ such that for any $\Sigma$-structures $\mathcal G_1,\mathcal G_2$ of degree at most $d$,
\begin{equation}
  \label{eq:goal}
  \mathcal G_1\foeq{f(\alpha)}\mathcal G_2\quad\text{entails}\quad\mathcal G_1\sieq{\alpha}\mathcal G_2\,.
\end{equation}

Indeed, this means that on the class of $\Sigma$-structures of degree at most $d$, any equivalence class $\mathcal C$ for $\sieq{\alpha}$ is a finite union of equivalence classes for $\foeq{f(\alpha)}$, and is consequently definable by an \FO-sentence $\varphi_{\mathcal C}$ of quantifier rank $f(\alpha)$. Let now $\mathcal P$ be a property of structures of degree at most $d$ definable by a sentence of \sifo of quantifier rank at most $\alpha$. It is a finite union $\bigcup_i\mathcal C_i$ of equivalence classes for $\sieq{\alpha}$. Hence, the $\FO[f(\alpha)]$-sentence $\bigvee_i\varphi_{\mathcal C_i}$ defines $\mathcal P$.

This proves the inclusion $\sifo\subseteq\FO$ on structures of degree at most $d$.

\bigskip

In order to prove (\ref{eq:goal}), we need $f(\alpha)$ to be large enough so as to enable us, given  \[\mathcal G_1\foeq{f(\alpha)}\mathcal G_2\,,\] to construct two successor relations $S_1$ and $S_2$ such that
\begin{equation}
  \label{eq:equiv_succ}
  (\mathcal G_1,S_1)\foeq{\alpha}(\mathcal G_2,S_2)\,,
\end{equation}
which in turn ensures that \[\quad\mathcal G_1\sieq{\alpha}\mathcal G_2\,.\]

\bigskip

Now, the Hanf threshold theorem yields two integers $r$ and $t$, depending on $\alpha$ and $d$, such that
\begin{equation}
  \label{eq:hanf}
  \threq{(\mathcal G_1, S_1)}{(\mathcal G_2, S_2)}{r}{t}
\end{equation}
entails (\ref{eq:equiv_succ}).

We have seen throughout this section how to construct two successor relations $S_1$ and $S_2$. All that remains is for us to show that, for the right value of $g$ in Lemma~\ref{lem:rare_freq}, our construction guarantees (\ref{eq:hanf}).

Let $\tau$ be an $r$-neighborhood type over $\Sigma\cup\{S\}$ which occurs in $(\mathcal G_1,S_1)$. There are two cases to consider:

\begin{itemize}
\item if $\tau$ occurs outside of $A_1$, then Lemma~\ref{lem:enriched_types} ensures that $\tau=\fractype{\chi}{r}$ for some frequent $r$-neighborhood type $\chi$. We can choose $g$ so that $\chi$ is guaranteed to have at least $t$ occurrences in $\mathcal G_1$ outside of $A_1$, and in $\mathcal G_2$ outside of $A_2$. This is ensured as long as \[g(\beta)\geq |A_1|+t\,,\] and in particular when \[g(\beta)\geq (\beta+2t_r^d)N(d+2,r)+t\,.\]
  
  Lemma~\ref{lem:enriched_types} then ensures that $\tau$ occurs at least $t$ times both in $(\mathcal G_1,S_1)$ and in $(\mathcal G_2,S_2)$.
  
\item if $\tau$ occurs only in $A_1$, then it cannot occur in $(\mathcal G_2,S_2)$ outside of $A_2$ (for the same reasons as above).  
  Lemma~\ref{lem:h_preserves_types} guarantees that $\tau$ has the same number of occurrences in $A_1$ and in $A_2$, hence in $(\mathcal G_1,S_1)$ and in $(\mathcal G_2,S_2)$.
\end{itemize}

As long as $g$ satisfies these conditions, (\ref{eq:hanf}) holds. We are now able to fix the value of $f(\alpha)$ as prescribed by Lemma~\ref{lem:rare_freq} for this $g$. This concludes the proof of Theorem~\ref{th:collapse}.

\section{Conclusion}

We have shown that \sifo collapses to \FO on any class of bounded degree.
In other words, we have shown that there exists a translation from \sifo to \FO on classes of bounded degree. As given by our proof, the quantifier rank of the translated sentence is triple-exponential in the quantifier rank of the original formula. It is an easy exercise to prove that the blowup is at least exponential, but we do not know if an exponential translation is at all possible.

Similar considerations arise when we take into account the length of the sentences instead of their quantifier rank - in this regard, our construction is even non-elementary, and all we know is that the blowup is at least exponential. 

An interesting task would be (i) to give an effective translation and (ii) to improve the succinctness of this translation, or to give tighter lower bounds on such translations. 

\medskip

Apart from these considerations, there are two main directions in which one could look to extend the present result.
One possibility would be to keep looking at classes of bounded degree while climbing up in the ladder of expressivity, and ask whether \oifo collapses to \FO as well on these classes of structures. 
New techniques would be needed, as contrary to what was the case with a successor, the addition of an order does not preserve the bounded degree property. Furthermore, even if $\oifo=\FO$ in this setting, it is not clear whether such orders can be directly constructed. It may be necessary to construct, as in \cite{DBLP:journals/jsyml/BenediktS09}, a chain of intermediate structures and orders.

\medskip

Alternatively, we could change the setting, and study the expressivity of \sifo on other sparse classes of structures, \eg, on classes of bounded treewidth. 
If showing the collapse of \sifo to \FO on these classes proved itself to be out of reach, a possibility would be to aim at proving that \sifo is Hanf-local (which would be stronger than the known Gaifman-locality). In that case, the starting hypothesis on the structures $\mathcal G_1$ and $\mathcal G_2$ would be stronger, as the existence of a $k$-neighborhood type-preserving bijection between the two structures would be assumed.

These tasks are much harder without any bound on the degree, which was what guaranteed that we could find elements of a given frequent neighborhood type far from each other.

\section*{Acknowledgment}
\noindent The author wishes to thank Luc Segoufin as well as the anonymous referees for their most appreciated feedback.

\bibliographystyle{alpha}
\bibliography{biblio} 

\end{document}